\documentclass[nonacm,acmsmall]{acmart}

\usepackage{preamble}

\tr{
\setcopyright{cc}
\setcctype[4.0]{by}
}{
\setcopyright{rightsretained}
\acmJournal{PACMMOD}
\acmYear{2024} \acmVolume{2} \acmNumber{N1 (SIGMOD)}
\acmArticle{2} \acmMonth{2} \acmPrice{15.00}
\acmDOI{10.1145/3639257}
}

\makeatletter
\gdef\@copyrightpermission{
  \begin{minipage}{0.2\columnwidth}
   \href{https://creativecommons.org/licenses/by/4.0/}{\includegraphics[width=0.90\textwidth]{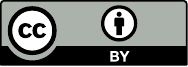}}
  \end{minipage}\hfill
  \begin{minipage}{0.8\columnwidth}
   \href{https://creativecommons.org/licenses/by/4.0/}{This work is licensed under a Creative Commons Attribution International 4.0 License.}
  \end{minipage}
  \vspace{5pt}
}
\makeatother



\AtBeginDocument{%
  \providecommand\BibTeX{{%
    \normalfont B\kern-0.5em{\scshape i\kern-0.25em b}\kern-0.8em\TeX}}}

\begin{document}
\tr{
\title{Optimizing Distributed Protocols with Query Rewrites [Technical Report]}
}{
\title{Optimizing Distributed Protocols with Query Rewrites}}

\author{David C. Y. Chu}
\orcid{0000-0001-9922-1994}
\affiliation{
    \institution{University of California, Berkeley}
    \country{USA}
}
\email{thedavidchu@berkeley.edu}

\author{Rithvik Panchapakesan}
\orcid{0009-0004-1428-5024}
\affiliation{
    \institution{University of California, Berkeley}
    \country{USA}
}
\email{rithvik@berkeley.edu}

\author{Shadaj Laddad}
\orcid{0000-0002-6658-6548}
\affiliation{
    \institution{University of California, Berkeley}
    \country{USA}
}
\email{shadaj@berkeley.edu}

\author{Lucky E. Katahanas}
\orcid{0009-0008-3073-0844}
\affiliation{
    \institution{Sutter Hill Ventures}
    \country{USA}
}
\email{lucky@shv.com}

\author{Chris Liu}
\orcid{0009-0002-1880-1941}
\affiliation{
    \institution{University of California, Berkeley}
    \country{USA}
}
\email{chris-liu@berkeley.edu}

\author{Kaushik Shivakumar}
\orcid{0009-0002-5943-9301}
\affiliation{
    \institution{University of California, Berkeley}
    \country{USA}
}
\email{kaushiks@berkeley.edu}

\author{Natacha Crooks}
\orcid{0000-0002-3567-801X}
\affiliation{
    \institution{University of California, Berkeley}
    \country{USA}
}
\email{ncrooks@berkeley.edu}

\author{Joseph M. Hellerstein}
\orcid{0000-0002-7712-4306}
\affiliation{
    \institution{University of California, Berkeley}
    \country{USA}
}
\affiliation{
    \institution{Sutter Hill Ventures}
    \country{USA}
}
\email{hellerstein@berkeley.edu}

\author{Heidi Howard}
\orcid{0000-0001-5256-7664}
\affiliation{
    \institution{Azure Research, Microsoft}
    \country{UK}
}
\email{heidi.howard@microsoft.com}

\renewcommand{\shortauthors}{David C. Y. Chu et al.}

\begin{abstract}
Distributed protocols such as 2PC and Paxos lie at the core of many systems in the cloud, but standard implementations do not scale.
New scalable distributed protocols are developed through careful analysis and rewrites, but this process is ad hoc and error-prone.
This paper presents an approach for scaling \emph{any} distributed protocol by applying rule-driven rewrites, borrowing from query optimization.
Distributed protocol rewrites entail a new burden: reasoning about spatiotemporal correctness.
We leverage order-insensitivity and data dependency analysis to systematically identify correct coordination-free scaling opportunities.
We apply this analysis to create preconditions and mechanisms for coordination-free decoupling and partitioning, two fundamental vertical and horizontal scaling techniques.
Manual rule-driven applications of decoupling and partitioning improve the throughput of 2PC by $5\times$ and Paxos by $3\times$, and match state-of-the-art throughput in recent work.
These results point the way toward automated 
optimizers for distributed protocols based on 
correct-by-construction rewrite rules.
\end{abstract}
\maketitle

\section{Introduction}
\label{sec:intro}


Promises of better cost and scalability have driven the migration of database systems to the cloud. 
Yet, the distributed protocols at the core of these systems, such as 2PC~\cite{mohan1986transaction} or Paxos~\cite{paxos}, are not designed to scale:
when the number of machines grows, overheads often increase and throughput drops.
As such, there has been a wealth of research on developing new, scalable distributed protocols.
Unfortunately, each new design requires careful examination of prior work and new correctness proofs; the process is ad hoc and often
error-prone~\cite{epaxos-broken,zyzzyvaBug, craqBug, raftBug, raftDissertation, protocolBugsList}.
Moreover, due to the heterogeneity
of proposed approaches, each new insight is localized to its particular protocol and cannot
easily be composed with other efforts.

This paper offers an alternative approach.
Instead of creating new distributed protocols from scratch, we formalize scalability optimizations into \textit{\changebars{rule-based}{rule-driven} rewrites} that are correct by construction 
and can be applied to \textit{any} distributed protocol.

To rewrite distributed protocols, we take a page from traditional SQL query optimizations.
Prior work has shown that distributed protocols
can be expressed declaratively as sets of queries in a SQL-like language such as Dedalus~\cite{dedalus}, which we adopt here. 
Applying query optimization to these protocols thus seems like an appealing way forward. Doing so correctly however, requires care, as the domain of distributed protocols requires optimizer transformations whose correctness is subtler than classical matters like the associativity and commutativity of join.
In particular, transformations to scale across machines must reason about program equivalence in the face of changes to spatiotemporal semantics like the order of data arrivals and the location of state.
We focus on applying two fundamental scaling optimizations in this paper: \emph{decoupling} and \emph{partitioning}, which correspond to vertical and horizontal scaling.
We target these two techniques because (1) they can be generalized across protocols and (2) were recently shown by Whittaker et al.~\cite{compartmentalized} to achieve state-of-the-art throughput on complex distributed protocols such as Paxos.
While Whittaker's rewrites are handcrafted specifically for Paxos, our goal is to \changebars{systematically decouple and partition}{rigorously define the general preconditions and mechanics for decoupling and partitioning, so they can be used to correctly rewrite} \emph{any} distributed protocol.



\textit{Decoupling} improves scalability by spreading \textit{logic} across machines to take advantage of additional physical resources and pipeline parallel computation.
Decoupling rewrites data dependencies on a single node into messages that are sent via asynchronous channels between nodes.
Without coordination, the original timing and ordering of messages cannot be guaranteed once these channels are introduced.
To preserve correctness without introducing coordination, we decouple sub-components that produce the same responses regardless of message ordering or timing: these sub-components are \emph{order-insensitive}.
Order-insensitivity is easy to systematically identify in Dedalus
thanks to its relational model: Dedalus programs are an (unordered) set of queries over (unordered) relations, so the logic for ordering---time, causality, log sequence numbers---is the exception, not the norm, and easy to identify.
By avoiding decoupling the logic that explicitly relies on order, we can decouple the remaining order-insensitive sub-components without coordination.

\textit{Partitioning} improves scalability by spreading \textit{state} across machines and parallelizing compute, a technique widely used in query processing~\cite{gamma,grace}.
Textbook discussions focus on partitioning data to satisfy a single query operator like join or group-by.
If the next operator downstream requires a different partitioning, then data must be forwarded or ``shuffled'' across the network.
We would like to partition data in such a way that \emph{entire sub-programs} can compute on local data without reshuffling.
We leverage relational techniques like functional dependency analysis to find data partitioning schemes that can allow as much code as possible to work on local partitions without reshuffling between operators.
This is a benefit of choosing to express distributed protocols in the relational model:
functional dependencies are far easier to identify in a relational language than a procedural language.

We demonstrate the generality of our optimizations by \changebars{systematically}{methodically} applying rewrites to three seminal distributed protocols: voting, 2PC, and Paxos.
We specifically target Paxos~\cite{paxosComplex} as it is a protocol with many distributed invariants and it is challenging to verify~\cite{verdi,ironfleet,distai}.
The throughput of the optimized voting, 2PC, and Paxos protocols scale by $2\times$, $5\times$, and $3\times$ respectively, a scale-up factor that matches the performance of \changebars{manual}{ad hoc} rewrites~\cite{compartmentalized} when the underlying language of each implementation is accounted for and achieves state-of-the-art performance for Paxos.

Our correctness arguments focus on the equivalence of localized,
``peephole'' optimizations of dataflow graphs.
Traditional protocol optimizations often make wholesale modifications to protocol logic and therefore require holistic reasoning to prove correctness.
We take a different approach. Our rewrite rules modify existing programs with small local changes, each of which is proven to preserve semantics. 
As a result, each rewritten subprogram is provably indistinguishable to an observer (or client) from the original.
We do not need to prove that holistic protocol invariants are preserved---they must be.
Moreover, because rewrites are local and preserve semantics, they can be \emph{composed} to produce protocols with multiple optimizations, as we demonstrate in~\Cref{sec:optimization-and-performance}.

Our local-first approach naturally has a potential cost: the space of protocol optimization is limited by design as it treats the initial implementation as ``law''.
It cannot distinguish between true protocol invariants and implementation artifacts, limiting the space of potential optimizations. Nonetheless, we find that, when applying our results to seminal distributed system algorithms, we easily match the results of their (manually proven) optimized implementations.

In summary, we make the following contributions:
\begin{enumerate}
    \item We present the preconditions and mechanisms for applying multiple correct-by-construction rewrites of two fundamental transformations: decoupling and partitioning.
    \item We demonstrate the application of these \changebars{rewrites in}{rule-driven rewrites by manually applying them to} complex distributed protocols such as Paxos.
    \item We evaluate our optimized programs and observe $2-5\times$ improvement in throughput across protocols with state-of-the-art throughput in Paxos,
    validating the role of correct-by-construction rewrites for distributed protocols.
\end{enumerate}

\tr{}{Due to a lack of space, the full precondition, mechanism, and proof of correctness of each rewrite in this paper can be found in the technical report~\cite{autocompTR}.}
\section{Background}
\label{sec:background}

Our contributions begin with the program rewriting rules in \Cref{sec:decoupling}.
Naturally, the correctness of those rules depends on the details of the language we are rewriting, Dedalus.
Hence in this section we pause to review the syntax and semantics of Dedalus, as well as additional terminology we will use in subsequent discussion.

Dedalus is a spatiotemporal logic for distributed systems~\cite{dedalus}.
As we will see in Section~\ref{sec:dedalus}, Dedalus captures specifications for the state, computation and messages of
a set of distributed \textbf{nodes} over time. Each node (a.k.a.\ machine, thread) has its own explicit ``clock'' that marks out local time sequentially.
Dedalus (and hence our work here) assumes a standard asynchronous model 
in which messages between correct nodes can be arbitrarily delayed and reordered, but must eventually be delivered after an infinite amount of time~\cite{dwork1988consensus}.

Dedalus is a dialect of Datalog$^\neg$, which is itself a SQL-like declarative logic language that supports familiar constructs like joins, selection, and projection, with additional support for recursion, aggregation (akin to \ded{GROUP BY} in SQL), and negation (\ded{NOT IN}). Unlike SQL, Datalog$^\neg$ has set semantics.

\subsection{Running example}
\label{sec:running-example}

\begin{figure}[t]
    \centering
    \includegraphics[width=0.5\linewidth]{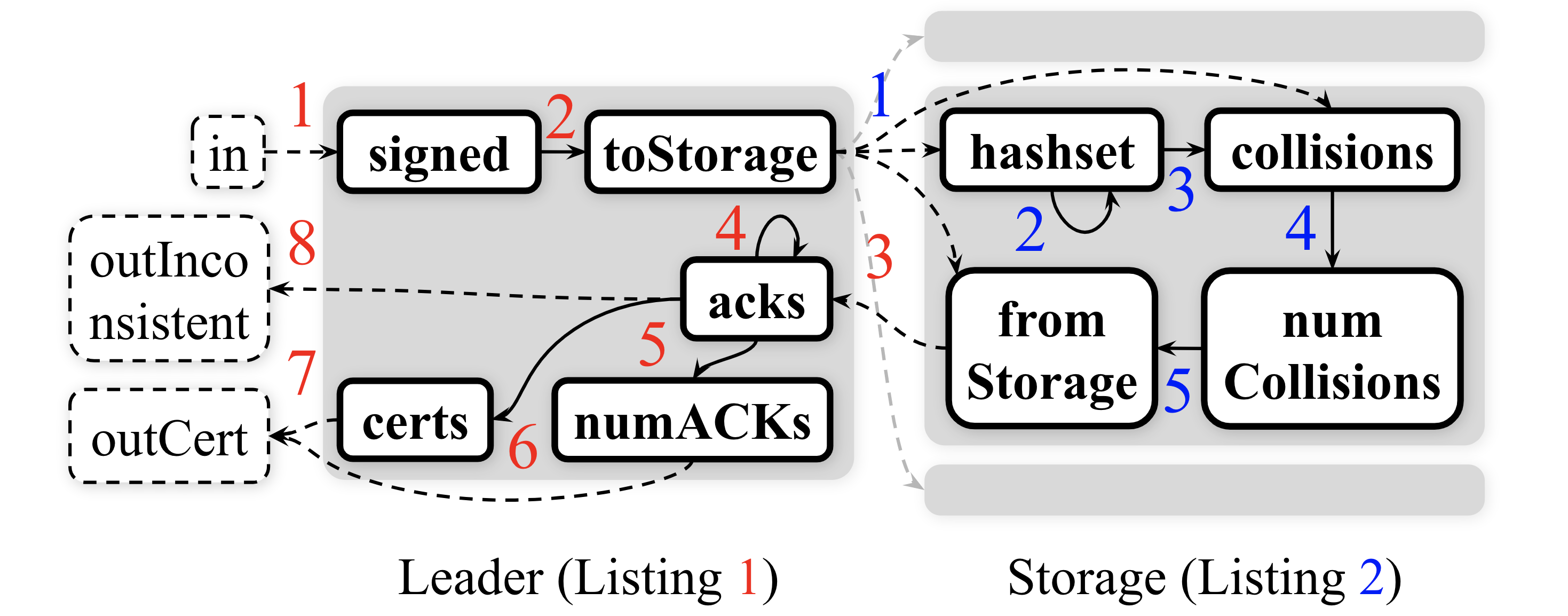}
    \caption{Dataflow diagram for a verifiably-replicated KVS.
    Edges are labeled with corresponding line numbers; dashed edges represent asynchronous channels. Each gray bounding box represents a node; select nodes' dataflows are presented.}
    \label{fig:running-example}
\end{figure}

As a running example, we focus on a verifiably replicated key-value store with hash-conflict detection inspired by~\cite{basil}.
We use this example to explain the core concepts of Dedalus and to illustrate in \Cref{sec:decoupling,sec:partitioning} how our transformations can be applied. In \Cref{sec:eval} we turn our attention to 
more complex and realistic examples, including Paxos and 2PC.
\Cref{fig:running-example} provides a high level diagram of the example; we explain the corresponding Dedalus code (\Cref{lst:storage,lst:leader}) in the next subsection.

The running example consists of a leader node and multiple storage nodes and allows clients to write to storage nodes, with the ability to detect concurrent writes.
The leader node cryptographically signs each client message and broadcasts both the message and signature to each storage node.
Each storage node then stores the message and the hash of the message in a local table if the signature is valid.
The storage nodes also calculate the number of unique existing messages in the table whose hash collides with the hash of the message.
The storage nodes then sign the original message and respond to the leader node.
Upon collecting a response from each storage node, if the number of hash collisions is consistent across responses, the leader creates a certificate of all the responses and replies to the client.
If any two storage nodes report differing numbers of hash collisions, the leader notifies the client of the inconsistency.
We use this simple protocol for illustration, and present more complete protocols---2PC and Paxos---in~\Cref{sec:eval}.

\subsection{Datalog$^\neg$}
\label{sec:datalog}
We now introduce the necessary Datalog$^\neg$ terminology,
copying code snippets from Listings~\ref{lst:leader} and~\ref{lst:storage} to introduce key concepts.


A Datalog$^\neg$ \textbf{program} is a set of \textbf{rules} in no particular order. 
A rule $\varphi$ is like a view definition in SQL, defining a virtual relation via a query over other relations.
A \textbf{literal} in a rule is either a relation, a negated relation, or a boolean expression.
A rule consists of a deduction operator \ded{:-} defining a single left-hand-side relation (the \textbf{head} of the rule) via a list of right-hand-side literals (the \textbf{body}).

Consider \Cref{line:storage-collisions} of \Cref{lst:storage}, which computes hash collisions:
\begin{lstlisting}[language=Dedalus, float=false, firstnumber=3]
collisions(val2,hashed,l,t) :- toStorage(val1,leaderSig,l,t), hash(val1,hashed), hashset(hashed,val2,l,t)
\end{lstlisting}
In this example, the head literal is \ded{collisions}, and the body literals are \ded{toStorage}, \ded{hash}, and \ded{hashset}.
Each body literal can be a (possibly negated) \textbf{relation} $r$ consisting of multiple \textbf{attributes} $A$, \emph{or}
a boolean expression; the head literal must be a relation. For example, \ded{hashset} is a relation with four attributes representing the hash, message value, 
location, and time in that order.
Each attribute must be bound to a constant or \textbf{variable}; attributes in the head literal can also be bound to \textbf{aggregation functions}.
In the example above, the attribute representing the message value in \ded{hashset} is bound to the variable \ded{val2}.
Positive literals in the body of the rule are joined together; negative literals are anti-joined (SQL's \ded{NOT IN}). 
Attributes bound to the same variable form an equality predicate---in the rule above, the first attribute of \ded{toStorage} must be equal to the first attribute of \ded{hash} since they are both bound to \ded{val1}; this specifies an equijoin of those two relations. 
Two positive literals in the same body that share no common variables form a cross-product.
Multiple rules may have the same head relation; the head relation is defined as the disjunction (SQL \ded{UNION}) of the rule bodies.

Note how library functions like \ded{hash} are simply modeled as infinite relations of the form \ded{(input, output)}.
Because these are infinite relations, they can only be used in a rule body if the input variables are bound to another attribute---this corresponds to ``lazily evaluating'' the function only for that attribute's finite set of values.
For example, the relation \ded{hash} contains the fact \ded{(x, y)} if and only if \ded{hash(x)} equals {y}.

Relations $r$ are populated with \textbf{facts} $f$, which are tuples of values, one for each attribute of $r$.
We will use the syntax $\pi_A(f)$ to project $f$ to the value of attribute $A$.
Relations with facts stored prior to execution are traditionally called \emph{extensional} relations, and the set of extensional relations is called the \textbf{EDB}. Derived relations, defined in the heads of rules, are traditionally called \emph{intensional} relations, and the set of them is called the \textbf{IDB}.
Boolean operators and library functions like \ded{hash} have pre-defined content, hence they are (infinite) EDB relations.


Datalog$^\neg$ also supports negation and aggregations. An example of aggregation is seen in \Cref{lst:storage} \Cref{line:storage-num-collisions}, which counts the number of hash collisions with the \ded{count} aggregation: 
\begin{lstlisting}[language=Dedalus, float=false, firstnumber=4]
numCollisions(count<val>,hashed,l,t) :- collisions(val,hashed,l,t)
\end{lstlisting}
In this syntax, attributes that appear outside of aggregate functions form the \ded{GROUP BY} list; attributes inside the functions are aggregated.
In order to compute aggregation in any rule $\varphi$, we must first compute the full content of all relations $r$ in the body of $\varphi$.
Negation works similarly: if we have a literal \ded{!r(x)} in the body, we can only check that \ded{r} is empty after
we're sure we have computed the full contents of \ded{r(x)}. 
We refer the reader to~\cite{aggAndNeg,alice} for further reading on aggregation and negation.



\subsection{Dedalus}
\label{sec:dedalus}
Dedalus programs are legal Datalog$^\neg$ programs, constrained to adhere to three additional rules on the syntax.

\textbf{(1) Space and Time in Schema:} All IDB relations must contain two attributes at their far right: location $L$ and time $T$.
Together, these attributes model \textit{where} and \textit{when} a fact exists in the system.
For example, in the rule on \Cref{line:storage-collisions} discussed above, a \ded{toStorage} message $m$ and signature $sig$ that arrives at time $t$ at a node with location $addr$ is represented as a fact \mbox{\ded{toStorage(}$m, sig, addr ,t$\ded{)}}.

\textbf{(2) Matching Space-Time Variables in Body:} 
The location and time attributes in \emph{all} body literals must be bound to the same variables $l$ and $t$, respectively. This models the physical property that two facts can be joined only if they exist at the same time and location.
In \Cref{line:storage-collisions}, a \ded{toStorage} fact that appears on node $l$ at time $t$ can only match with \ded{hashset} facts that are also on $l$ at time $t$. 

We model library functions like \ded{hash} as relations
that are known (replicated) across all nodes $n$ and unchanging across all timesteps $t$. Hence we elide $L$ and $T$ from function and expression literals as a matter of syntax sugar, and assume they can join with other literals at all locations and times.

\textbf{(3) Space and Time Constraints in Head:} The location and time variables in the \emph{head} of rules must obey certain syntactic constraints, which 
ensure that the ``derived'' locations and times correspond to physical reality. 
These constraints differ across three types of rules.
\textbf{Synchronous} (``deductive''~\cite{dedalus}) 
rules are captured by having the same time variable in the head literal as in the body literals. Having these derivations assigned to the same timestep $t$ is only physically possible on a single node, so the location in the head of a synchronous rule must match the body as well. 
\textbf{Sequential} (``inductive''~\cite{dedalus}) rules are captured by having the head literal's time be the successor (\ded{t+1}) of the body literals' times \ded{t}.
Again, sequentiality can only be guaranteed physically on a single node in an asychronous system, so the location of the head in a sequential rule must match the body.
\textbf{Asynchronous} rules capture message passing between nodes, by having different
time and location variables in the head than the body. In an asynchronous system, messages are delivered at an arbitrary time in the future.
We discuss how this is modeled next. 

In an asynchronous rule $\varphi$,
the location attribute of the head and body relations in $\varphi$ are bound to different variables; a different location in the head of $\varphi$ indicates the arrival of the fact on a new node.
Asynchronous rules are constrained to capture non-deterministic delay by including a body literal for the built-in \ded{delay} relation (a.k.a. \ded{choose}~\cite{dedalus}, \ded{chosen}~\cite{dedalusSemantics}), a non-deterministic function that independently maps each  head fact to an arrival time. 
The logical formalism of the \ded{delay} function is discussed in~\cite{dedalusSemantics}; for our purposes it is sufficient to know that \ded{delay} is constrained to reflect Lamport's ``happens-before'' relation for each fact. That is, a fact sent at time $t$ on $l$ arrives at time $t'$ on $l'$, where $t < t'$.
We focus on 
\Cref{lst:storage}, \Cref{line:storage-ACK} from our running example.
\begin{lstlisting}[language=Dedalus, float=false, firstnumber=5]
fromStorage(l,sig,val,collCnt,l',t') :- toStorage(val,leaderSig,l,t), hash(val,hashed), numCollisions(collCnt,hashed,l,t), sign(val,sig), leader(l'), delay((sig,val,collCnt,l,t,l'),t')
\end{lstlisting}
This is an asynchronous rule where a storage node $l$ sends the count of hash collisions for each distinct storage request back to the leader $l'$.
Note the \ded{l'} and \ded{t'} in the head literal: they are derived from the body literals \ded{leader} (an EDB relation storing the leader address) and the built-in \ded{delay}. 
Note also how the first attribute of \ded{delay} (the function ``input'') is a tuple of variables that, together, distinguish each individual head fact. This allows \ded{delay} to choose a different \ded{t'} for every head fact~\cite{dedalusSemantics}.
The \ded{l} in the head literal represents the storage node's address and is used by the leader to count the number of votes; it is unrelated to asynchrony.

So far, we have only talked about facts that exist at a point in time $t$.
State change in Dedalus is modeled through the existence or non-existence of facts \emph{across} time.
\textbf{Persistence rules} like the one below from \Cref{line:storage-persist} of Listing~\ref{lst:storage} ensure, inductively, that facts in \ded{hashset} that exist at time $t$ exist at time $t+1$.
Relations with persistence rules---like \ded{hashset}---are \textbf{persisted}.
\begin{lstlisting}[language=Dedalus, float=false, firstnumber=2]
hashset(hashed,val,l,t') :- hashset(hashed,val,l,t), t'=t+1
\end{lstlisting}

\subsection{Further terminology}
\label{sec:further-terminology}

We introduce some additional terminology to capture 
the rewrites we wish to perform on Dedalus programs.

We assume that Dedalus programs are composed of separate \textbf{components} $C$, each with a non-empty set of rules $\overline{\varphi}$.
In our running example, \Cref{lst:leader,lst:storage} define the leader component and the storage component.
All the rules of a component are executed together on a single physical node.
Many instances of a component may be deployed, each on a different node.
The node at location \ded{addr} only has access to facts $f$ with $\pi_L(f) =$ \ded{addr}, modeling 
the shared-nothing property of distributed systems.

We define a rule's \textbf{references} as the IDB relations in its body; a component
references the set of relations referenced by its rules.
For example,
the storage component in Listing~\ref{lst:storage} references \ded{toStorage}, \ded{hashset}, \ded{collisions}, and \ded{numCollisions}.
A IDB relation is an \textbf{input} of a component $C$ if it is referenced in $C$ and it is not in the head of any rules of $C$; \ded{toStorage} is an input to the storage component.
A relation that is not referenced in $C$ but appears in the head of rules in $C$ is an \textbf{output} of $C$; \ded{fromStorage} is an output of the storage component. 
Note that this formulation explicitly allows a component to have multiple inputs and multiple outputs.
Inputs and outputs of the component correspond to asynchronous input and output channels of each node.

Our discussion so far has been at the level of rules; we will also need to reason about individual facts. 
A \textbf{proof tree}~\cite{alice} can be constructed for each IDB fact $f$, where $f$ lies at the root of the tree,
each leaf is an EDB or input fact, and each internal node is an IDB fact derived from its children via a single rule.
Below we see
a proof tree for one fact in \ded{toStorage}:

\begin{tikzpicture}[
every node/.style={fill=lightgray},
level 1/.style={sibling distance=45mm, level distance=0.8cm},
level 2/.style={sibling distance=40mm, level distance=0.8cm},
level 3/.style={sibling distance=20mm},
mylabel/.style={draw=none, fill=none, text=gray, font=\footnotesize, inner sep=0pt}
]
\node (root) {\fact{toStorage('hi', 0x7465, b.b.us:5678, 9)}} {
child { node (signed) {\fact{signed('hi', 0x7465, a.b.us:5678, 6)}} {
child { node[right] {\fact{in('hi', a.b.us:5678, 6)}}}
child { node[right] {\fact{sign('hi', 0x7465)}}}
}}
child { node {\fact{storageNodes(b.b.us:5678)}} }
child { node {\fact{delay(('hi', 0x7465, a.b.us:5678, 6, b.b.us:5678), 9)}} }
};

\node[mylabel, above right=-1mm and 6mm of root] (label1) {{\tiny \Cref{line:leader-broadcast}}};
\draw[dashed, ->, gray] (label1) to[bend right=20] (root);

\node[mylabel, above left=2mm and -8mm of signed] (label2) {{\tiny \Cref{line:leader-sign}}};
\draw[dashed, ->, gray] (label2) to[bend left=30] (signed);
\end{tikzpicture}




\begin{lstlisting}[language=Dedalus, label={lst:leader}, caption={Hashset leader in Dedalus.}]
signed(val,leaderSig,l,t) :- in(val,l,t), sign(val,leaderSig) |\label{line:leader-sign}|
toStorage(val,leaderSig,l',t') :- signed(val,leaderSig,l,t), storageNodes(l'), delay((val,leaderSig,l,t,l'),t') |\label{line:leader-broadcast}|
acks(src,sig,val,collCnt,l,t) :- fromStorage(src,sig,val,collCnt,l,t) |\label{line:leader-ACK}|
acks(src,sig,val,collCnt,l,t') :- acks(src,sig,val,collCnt,l,t), t'=t+1 |\label{line:leader-persist}|
numACKs(count<src>,val,collCnt,l,t) :- acks(src,sig,val,collCnt,l,t) |\label{line:leader-count-ACKs}|
certs(cert<sig>,val,collCnt,l,t) :- acks(src,sig,val,collCnt,l,t) |\label{line:leader-cert}|
outCert(cer,val,collCnt,hashed,l',t') :- certs(ce,val,collCnt,l,t), numACKs(cnt,val,collCnt,l,t), numNodes(cnt), client(l'), delay((cer,val,collCnt,hashed,l,t,l'),t') |\label{line:leader-cert-out}|
outInconsistent(val,l',t') :- acks(src1,sig1,val,collCnt1,l,t), acks(src2,sig2,val,collCnt2,l,t), collCnt1 != collCnt2, client(l'), delay((val,l,t,l'),t') |\label{line:leader-inconsistency-out}|
\end{lstlisting}

\begin{lstlisting}[language=Dedalus, label={lst:storage}, caption={Hashset storage node in Dedalus.}]
hashset(hashed,val,l,t') :- toStorage(val,leaderSig,l,t), hash(val,hashed), verify(val,leaderSig), t'=t+1 |\label{line:storage-write}|
hashset(hashed,val,l,t') :- hashset(hashed,val,l,t), t'=t+1 |\label{line:storage-persist}|
collisions(val2,hashed,l,t) :- toStorage(val1,leaderSig,l,t), hash(val1,hashed), hashset(hashed,val2,l,t) |\label{line:storage-collisions}|
numCollisions(count<val>,hashed,l,t) :- collisions(val,hashed,l,t) |\label{line:storage-num-collisions}|
fromStorage(l,sig,val,collCnt,l',t') :- toStorage(val,leaderSig,l,t), hash(val,hashed), numCollisions(collCnt,hashed,l,t), sign(val,sig), leader(l'), delay((sig,val,collCnt,l,t,l'),t') |\label{line:storage-ACK}|
\end{lstlisting}

\subsection{Correctness}
\label{sec:general-correctness}
This paper transforms single-node Dedalus components into ``equivalent'' multi-component, multi-node Dedalus programs; the transformations can be composed to scale entire distributed protocols.
For equivalence, we want a definition that satisfies any client (or observer) of the input/output channels of the original program.
To this end we employ equivalence of concurrent histories as defined for linearizability~\cite{linearizability}, the gold standard in distributed systems.


We assume that a history $H$ can be constructed from any run of a given Dedalus program $P$. Linearizability traditionally expects every program to include a specification that defines what histories are "legal". We make no such assumption and we consider any possible history generated by the unoptimized program $P$ to define the specification.  As such, the optimized program $P'$ is linearizable if any run of $P'$ generates the same output facts with the same timestamps as some run of $P$.

Our rewrites are safe over protocols that assume the following fault model:
an asynchronous network (messages between correct nodes will eventually be delivered) where up to $f$ nodes can suffer from
general omission failures~\cite{generalOmission} (they may fail to send or receive some messages).
After optimizing, one original node $n$ may be replaced by multiple nodes $n_1, n_2, \ldots$; the failure of any of nodes $n_i$ corresponds to a partial failure of the original node $n$, which is equivalent to the failure of $n$ under general omission.


\tr{Full proofs, preconditions, and mechanisms for the rewrites described in \Cref{sec:decoupling,sec:partitioning} can be found in \Cref{app:decoupling}.}
{Due to a lack of space, we omit the proofs of correctness of the rewrites described in~\Cref{sec:decoupling,sec:partitioning}.
Full proofs, preconditions, and rewrite mechanisms can be found in the appendix of our technical report~\cite{autocompTR}.}
\section{Decoupling}
\label{sec:decoupling}

Decoupling partitions code; it takes a Dedalus component running on a single node, and breaks it into multiple components that can run in parallel across many nodes.
Decoupling can be used to alleviate single-node bottlenecks by scaling up available resources.
Decoupling can also introduce pipeline parallelism: if one rule produces facts in its head that another rule consumes in its body, decoupling those rules across two components can allow the producer and consumer to run in parallel.

Because Dedalus is a language of unordered rules, decoupling a component is syntactically easy: we simply partition
the component's ruleset into multiple subsets, and assign each subset to a different node.
The result is syntactically legal, but the correctness story is not quite that simple.
To decouple \emph{and} retain the original program semantics, we must address
classic distributed systems challenges: how to get the right data to the right nodes (space), and how 
to ensure that introducing asynchronous messaging between nodes does not affect correctness (time).


In this section we step through a progression of decoupling scenarios, and introduce analyses and rewrites that provably address our concerns regarding space and time. Throughout, our goal is to avoid introducing
any \emph{coordination}---i.e. extra messages beyond the data passed between rules in the original program. 

\noindent
\textbf{General Construction for Decoupling:}
In all our scenarios we will consider a component $C$ at network location \ded{addr}, consisting of a set of rules $\overline{\varphi}$.
We will, without loss of generality, decouple $C$ into two components: $C_1 = \overline{\varphi}_1$, which stays at location \ded{addr}, 
and $C_2 = \overline{\varphi}_2$ which is placed at a new location \ded{addr2}.
The rulesets of the two new components partition the original ruleset: $\overline{\varphi}_1 \cap \overline{\varphi}_2 = \emptyset$ and 
$\overline{\varphi}_1 \cup \overline{\varphi}_2 \supseteq \overline{\varphi}$. 
Note that we may add new rules during decoupling to 
achieve equivalence.



\subsection{Mutually Independent Decoupling}
\label{sec:mutually-independent-decoupling}

Intuitively, if the component $C_1$ never communicates with $C_2$, then running them on two separate nodes should not change program semantics.
We simply need to ensure that inputs from other components are sent to \ded{addr} or \ded{addr2} appropriately.

\begin{figure}[t]
    \centering
    \includegraphics[width=0.5\linewidth]{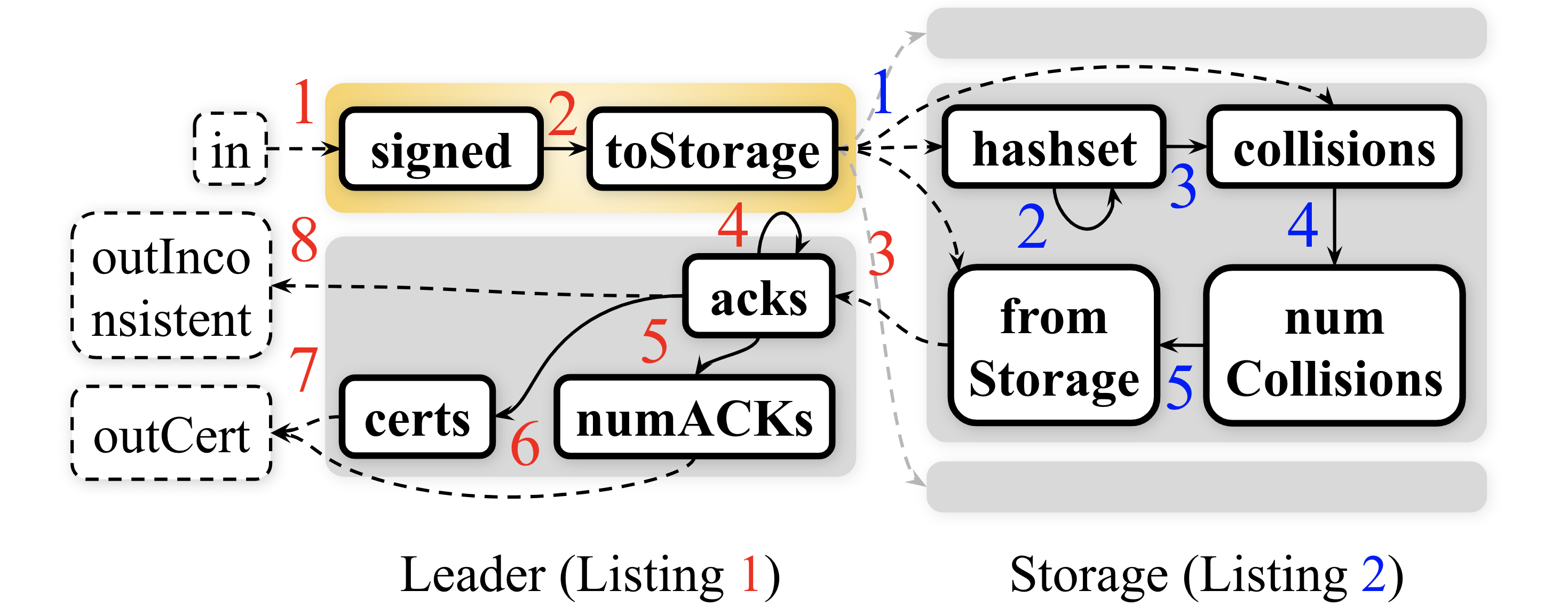}
    \caption{Running example after mutually independent decoupling.}
    \label{fig:running-example-mutually-independent-decoupling}
\end{figure}


Consider the component defined in \Cref{lst:leader}.
There is no dataflow between the relations in \Cref{line:leader-sign,line:leader-broadcast} and the relations in the remainder of the rules in the component.
One possible decoupling would place \Cref{line:leader-sign,line:leader-broadcast} on $C_1$, the remainder of \Cref{lst:leader} on $C_2$, and reroute \ded{fromStorage} messages from $C_1$ to $C_2$, as seen in \Cref{fig:running-example-mutually-independent-decoupling}.

We now define a precondition that determines when this rewrite can be applied:

\noindent
\textbf{Precondition:} 
$C_1$ and $C_2$ are mutually independent.

Recall the definition of \emph{references} from \Cref{sec:further-terminology}:
a component $C$ references IDB relation $r$ if some rule $\varphi \in C$ has $r$ in its body.
A component $C_1$ is \textit{independent} of component $C_2$ if
(a) the two components reference mutually exclusive sets of relations, and
(b) $C_1$ does not reference the outputs of $C_2$.
Note that this property is asymmetric:
$C_2$ may still be dependent upon $C_1$ by referencing $C_1$'s outputs. Hence our precondition requires \emph{mutual} independence.

\noindent
\textbf{Rewrite: Redirection.} 
Because $C_2$ has changed address, we need to direct facts from any relation $r$ referenced by $C_2$ to \ded{addr2}.
We simply add a ``redirection'' EDB relation to the body of each rule 
whose head is referenced in $C_2$, which maps \ded{addr} to \ded{addr2}, and any other address to itself. For our example above, we need to ensure that \ded{fromStorage} is sent 
to \ded{addr2}. To enforce this we rewrite \Cref{line:storage-ACK} of \Cref{lst:storage} as follows (note variable \ded{l''} in the head, and \ded{forward} in the body):
\begin{lstlisting}[language=Dedalus, float=false, firstnumber=5]
fromStorage(l,sig,val,collCnt,l'',t') :- toStorage(val,leaderSig,l,t), hash(val,hashed), numCollisions(collCnt,hashed,l,t), sign(val,sig), leader(l'), forward(l',l'') delay((l,sig,val,collCnt,l,t,l''),t')
\end{lstlisting}





\subsection{Monotonic Decoupling}
\label{sec:monotonic-decoupling}

Now consider a scenario in which $C_1$ and $C_2$ are not mutually independent.
If $C_2$ is dependent on $C_1$, decoupling changes the dataflow from $C_1$ to $C_2$ 
to traverse asynchronous channels.
After decoupling, facts that co-occurred in $C$ may be spread across time in $C_2$; similarly, two facts that were ordered or timed in a particular way in $C$ may be ordered or timed differently in $C_2$.
Without coordination, very little can be guaranteed about the behavior of a component after the ordering or timing of facts is modified.


Fortunately, the CALM Theorem~\cite{calm} tells us that \textit{monotonic} components eventually produce the same output independent of any network delays, including changes to co-occurrence, ordering, or timing of inputs. 
A component $C_2$ is monotonic if increasing its input set from $I$ to $I' \supseteq I$ implies that the output set $C_2(I') \supseteq C_2(I)$\footnote{There is some abuse of notation here treating $C_2$ as a function from one set of facts to to another, since the facts may be in different relations.
A more proper definition would be based on sets of multiple relations: input and EDB relations at the input, IDB relations at the output.};
in other words, each referenced relation and output of $C_2$ will monotonically accumulate a growing set of facts as inputs are received over time, independent of the order in which they were received.
The CALM Theorem ensures that if $C_2$ is shown to be monotonic, then we can safely decouple $C_1$ and $C_2$ without any coordination.


\begin{figure}[t]
    \centering
    \includegraphics[width=0.5\linewidth]{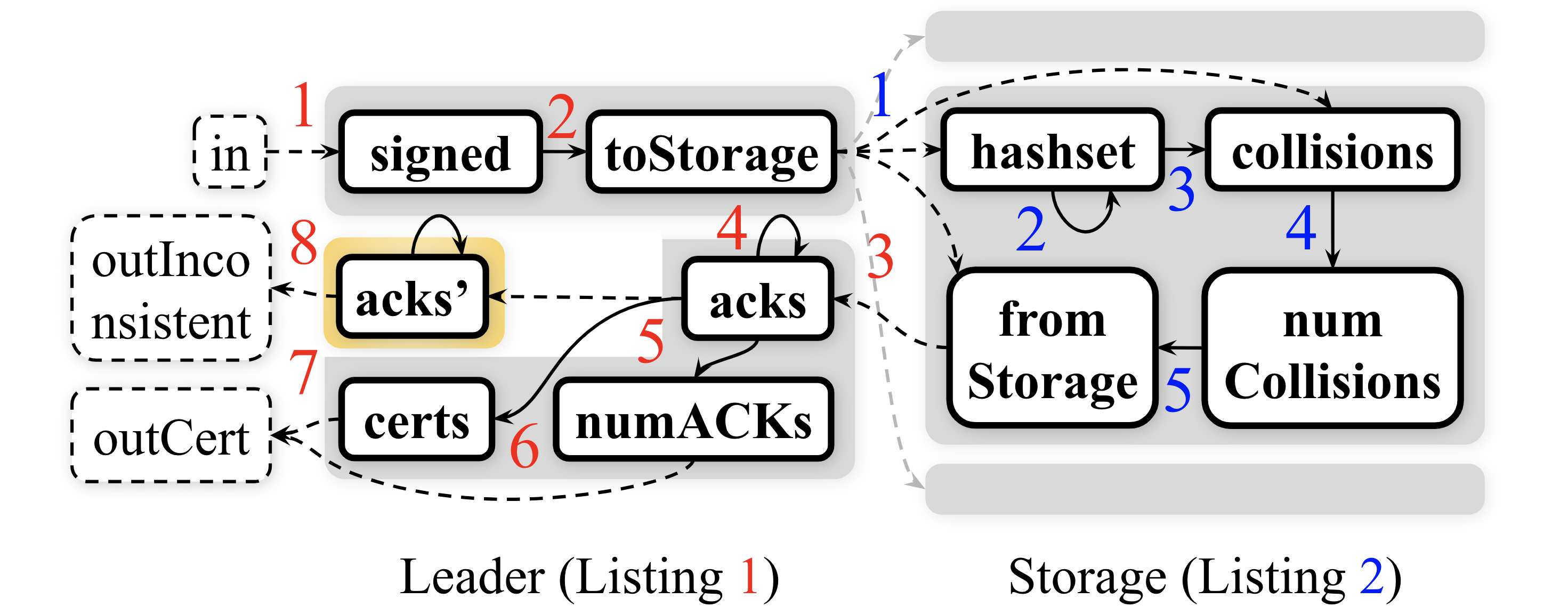}
    \caption{Running example after monotonic decoupling.}
    \label{fig:running-example-monotonic-decoupling}
\end{figure}

In our running example, the leader (\Cref{lst:leader}) is responsible for both creating certificates from a set of signatures (\Cref{line:leader-count-ACKs,line:leader-cert,line:leader-cert-out}) and checking for inconsistent ACKs (\Cref{line:leader-inconsistency-out}).
Since ACKs are persisted, once a pair is inconsistent, they will always be inconsistent; \Cref{line:leader-inconsistency-out} is monotonic.
Monotonic decoupling of \Cref{line:leader-inconsistency-out} allows us to
offload inconsistency-checking from a single leader to the decoupled ``proxy'' as highlighted in yellow in \Cref{fig:running-example-monotonic-decoupling}.

\noindent
\textbf{Precondition:}
$C_1$ is independent of $C_2$, and $C_2$ is \textbf{monotonic}.

Monotonicity of a Datalog$^\neg$ (hence Dedalus) component is undecidable~\cite{monotonicityUndecidable}, but effective conservative tests for monotonicity are well known.
A simple sufficient condition for monotonicity is to ensure that (a) $C_2$'s input relations are persisted,
and (b) $C_2$'s rules do not contain negation or aggregation.
In \tr{\Cref{app:monotonic-decoupling}}{the technical report} we relax each of these checks to be more permissive.

\noindent
\textbf{Rewrite: Redirection With Persistence.} 
Note that in this case we may have relations $r$ that are outputs of $C_1$ and inputs to $C_2$.
We use the same rewrite as in the previous section with one addition: we add a persistence rule to $C_2$ for each $r$ that is in the output of $C_1$ and the input of $C_2$, guaranteeing that all inputs of $C_2$ remain persisted.


The alert reader may notice performance concerns. First, $C_1$ may redundantly resend persistently-derived facts to $C_2$ each tick, even though $C_2$ is persistently storing them anyway via the rewrite.
Second, $C_2$ is required to persist facts indefinitely, potentially long after they are needed.
Solutions to this problem were explored in prior work~\cite{edelweiss} and can be incorporated here as well without affecting semantics.



\subsection{Functional Decoupling}
\label{sec:order-insensitive-decoupling}
Consider a component that behaves like a ``map'' operator for a pure function $F$ on individual facts: for each fact $f$ it receives as input, it outputs $F(f)$.
Surely these should be easy to decouple! 
Map operators are monotonic (their output set grows with their input set),
but they are also independent per fact---each output is determined only by its corresponding input, and in particular is not affected by previous inputs.
This property allows us to forgo the persistence rules we introduce for more general monotonic decoupling; we refer to this special case of monotonic decoupling as \emph{functional decoupling}.

\begin{figure}[t]
    \centering
    \includegraphics[width=0.5\linewidth]{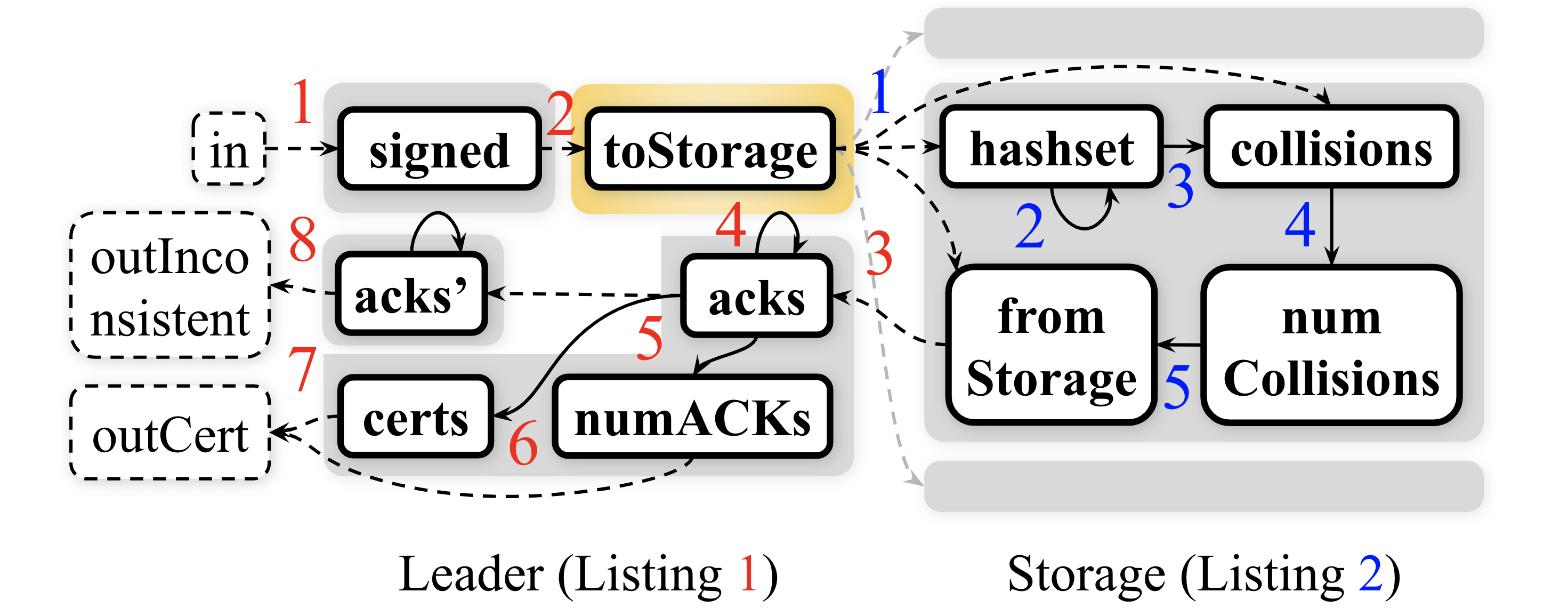}
    \caption{Running example after functional decoupling.}
    \label{fig:running-example-functional-decoupling}
\end{figure}

Consider again \Cref{line:leader-sign,line:leader-broadcast} in \Cref{lst:leader}.
Note that \Cref{line:leader-sign} works like a function on one input: each fact from \ded{in} results in an independent signed fact in \ded{signed}.
Hence we can decouple further, placing \Cref{line:leader-sign} on one node and \Cref{line:leader-broadcast} on another, forwarding signed values to \ded{toStorage}.
Intuitively, this decoupling does not change program semantics because \Cref{line:leader-broadcast} simply sends messages, regardless of which messages have come before: it behaves like pure functions.

\noindent
\textbf{Precondition:} $C_1$ is independent of $C_2$, and $C_2$ is \textbf{functional}---that is, (1) it does not contain aggregation or negation, and (2) each rule body in $C_2$ has at most one IDB relation.

\noindent
\textbf{Rewrite: Redirection.} We reuse the rewrite from Section~\ref{sec:mutually-independent-decoupling}.


As a side note, recall that persisted relations in Dedalus are by definition IDB relations.
Hence Precondition (2) prevents $C_2$ from joining current inputs (an IDB relation) with previous persisted data (another IDB relation)!
In effect, persistence rules are irrelevant to the output of a functional component, rendering functional components effectively ``stateless''.

\section{Partitioning}
\label{sec:partitioning}

Decoupling is the distribution of \emph{logic} across nodes;
partitioning (or ``sharding'') is the distribution of \emph{data}.
By using a relational language like Dedalus, we
can scale protocols using a variety of techniques that query optimizers use to maximize partitioning without excessive ``repartitioning'' (a.k.a. ``shuffling'') of data at runtime.

Unlike decoupling, which introduces new components, partitioning introduces additional nodes on which to run instances of each component. 
Therefore, each fact may be rerouted to any of the many nodes, depending on the partitioning scheme.
Because each rule still executes locally on each node, we must reason about changing the \emph{location} of facts.

We first need to define partitioning schemes, and what it means for a partitioning to be correct for a set of rules.
Much of this can be borrowed from recent theoretical literature~\cite{distConstraints,distSchemes,parallelCorrectness, parallelTransfer}.
A partitioning scheme is described by a \textit{distribution policy} $D(f)$ that outputs some node address \ded{addr\_i} for any fact $f$.
A partitioning 
preserves the semantics of the rules in a component
if it is \textbf{parallel disjoint correct}~\cite{distSchemes}.
Intuitively, this property says that the body facts that need to be colocated remain colocated after partitioning.
We adapt the parallel disjoint correctness definitions to the context of Dedalus as follows:

\begin{definition}
A distribution policy $D$ over component $C$ is \textit{parallel disjoint correct} if for any fact $f$ of $C$, for any two facts $f_1,f_2$ in the proof tree of $f$, $D(f_1) = D(f_2)$.
\end{definition}


Ideally we can find a single distribution policy that is parallel disjoint correct over the component in question. To do so,
we need to partition each relation based on the set of attributes used for joining or grouping the relation in the component's rules.
Such distribution policies are said to satisfy the co-hashing constraint (\Cref{sec:co-hashing-partitioning}).
Unfortunately, it is common for a single relation to be referenced in two rules with different join or grouping attributes.
In some cases, dependency analysis can still find a distribution policy that will be correct~(\Cref{sec:dependencies}).
If no parallel disjoint correct distribution policy can be found, we can resort to partial partitioning (\Cref{sec:partial-partitioning}), which replicates facts across multiple nodes.

To discuss partitioning rewrites on generic Dedalus programs, we consider without loss of generality a component $C$ with a set of rules $\overline{\varphi}$ at network location \ded{addr}.
We will partition the data at \ded{addr} across a set of new locations \ded{addr1}, \ded{addr2}, etc, each executing the same rules $\overline{\varphi}$.

\subsection{Co-hashing}
\label{sec:co-hashing-partitioning}

We begin with co-hashing~\cite{distSchemes,distConstraints}, a well studied constraint that avoids repartitioning data. 
Our goal is to co-locate facts that need to be combined because they (a) share a join key, (b) share a group key, or (c) share an antijoin key. 

Consider two relations $r_1$ and $r_2$ that appear in the body of a rule $\varphi$, with matching variables bound to attributes $A$ in $r_1$ and corresponding attributes $B$ in $r_2$.
Henceforth we will say that $r_1$ and $r_2$ ``share keys'' on attributes $A$ and $B$.
Co-hashing states that if $r_1$ and $r_2$ share keys on attributes $A$ and $B$, then all facts from $r_1$ and $r_2$ with the same values for $A$ and $B$ must be routed to the same partition.

\begin{figure}[t]
    \centering
    \includegraphics[width=0.5\linewidth]{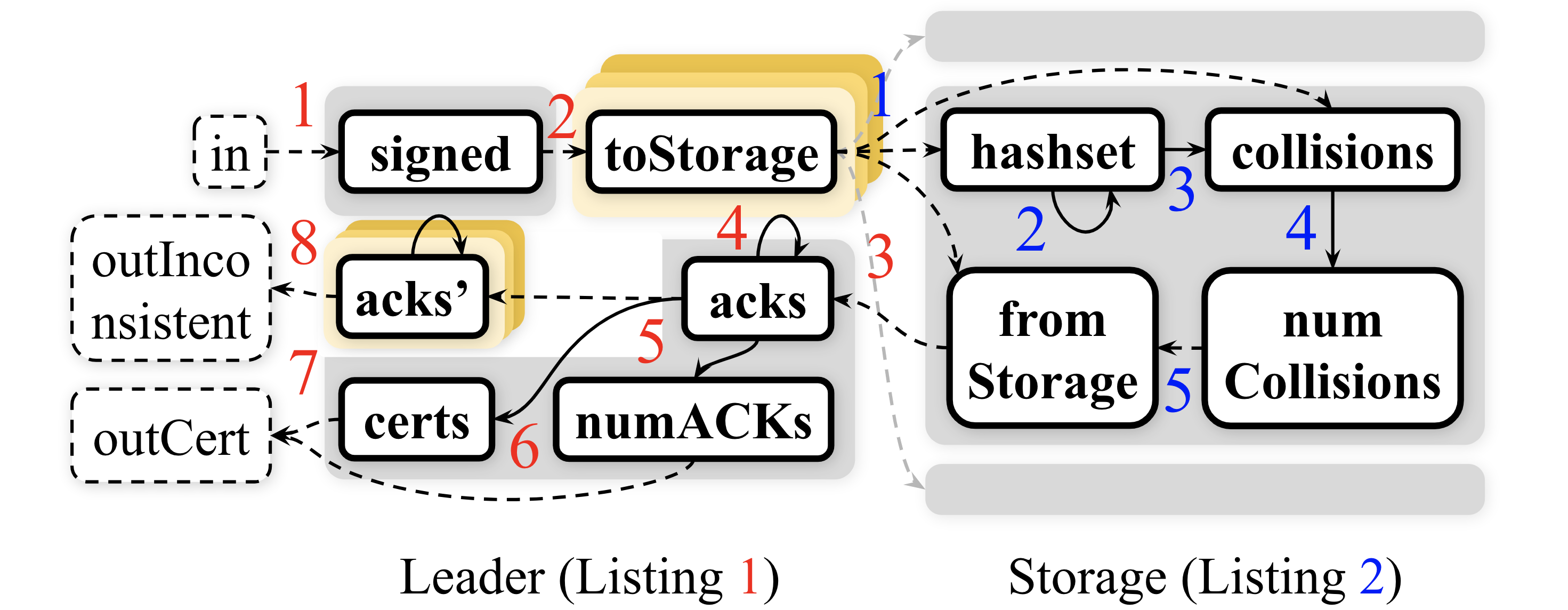}
    \caption{Running example after partitioning with co-hashing.}
    \label{fig:running-example-partitioning-with-co-hashing}
\end{figure}


Note that even if co-hashing is satisfied for individual rules, $r$ might need to be repartitioned \emph{between} the rules, because a relation $r$ might share keys with another relation on attributes $A$ in one rule and $A'$ in another.
To avoid repartitioning, we would like the distribution policy to partition consistently with co-hashing in \emph{every} rule of a component.

Consider \Cref{line:leader-inconsistency-out} of \Cref{lst:leader}, assuming it has already been decoupled.
Inconsistencies between ACKs are detected on a per-value basis and can be partitioned over the attribute bound to the variable \ded{val}; this is evidenced by the fact that the relation \ded{acks} is always joined with other IDB relations using the same attribute (bound to \ded{val}).
\Cref{line:leader-broadcast} and \Cref{lst:storage} \Cref{line:storage-ACK} are similarly partitionable by value, as seen in \Cref{fig:running-example-partitioning-with-co-hashing}.

Formally, a distribution policy 
$D$ \textbf{partitions} relation $r$ by attribute $A$ if for any pair of facts $f_1,f_2$ in $r$, $\pi_A(f_1) = \pi_A(f_2)$ implies $D(f_1) = D(f_2)$.
Facts are distributed according to their partitioning attributes.

$D$ \textbf{partitions consistently with co-hashing} if for any pair of referenced relations $r_1,r_2$ in rule $\varphi$ of $C$, $r_1$ and $r_2$ share keys on attribute lists $A_1$ and $A_2$ respectively, such that for any pair of facts $f_1 \in r_1, f_2 \in r_2$, $\pi_{A_1}(f_1) = \pi_{A_2}(f_2)$ implies $D(f_1) = D(f_2)$.
Facts will be successfully joined, aggregated, or negated after partitioning because they are sent to the same locations.

\noindent
\textbf{Precondition:} 
There exists a distribution policy $D$ for relations referenced by component $C$ that partitions consistently with co-hashing.

We can discover candidate distribution policies through a static analysis of the join and grouping attributes in every rule $\varphi$ in $C$.

\noindent
\textbf{Rewrite: Redirection With Partitioning.}
We are given a distribution policy $D$ from the precondition.
For any rules in $C'$ whose head is referenced in $C$, we modify the ``redirection'' relation such that messages $f$ sent to $C$ at \ded{addr} are instead sent to the appropriate node of $C$ at $D(f)$. 



\subsection{Dependencies}
\label{sec:dependencies}



By analyzing Dedalus rules, we can identify dependencies between attributes that (1) strengthen partitioning by showing that partitioning on one attribute can imply partitioning on another, and (2) loosen the co-hashing constraint.

For example, consider a relation $r$ that contains both an original string attribute \ded{Str} and its uppercased value in attribute \ded{UpStr}.
The \textbf{functional dependency} (FD) $\ded{Str} \to \ded{UpStr}$ strengthens partitioning: partitioning on \ded{UpStr} implies partitioning on \ded{Str}.
Formally, relation $r$ has a \textit{functional dependency} $g: A \to B$ on attribute lists $A,B$ if for all facts $f \in r$, $\pi_B(f) = g(\pi_A(f))$ for some function $g$. That is, the values $A$ in the domain of $g$ determine the values in the range, $B$.
This reasoning allows us to satisfy multiple co-hashing constraints simultaneously.

Now consider the following joins in the body of a rule: \ded{p(str), r(str, upStr), q(upStr)}.
Co-hashing would not allow partitioning, because $p$ and $q$ do not share keys over their attributes.
However, if we know the functional dependency $\ded{Str} \rightarrow \ded{UpStr}$ over $r$, then we can partition $p,q,r$ on the uppercase values of the strings and still avoid reshuffling.
This \textbf{co-partition dependency} (CD) between the attributes of $p$ and $q$ loosens the co-hashing constraint beyond sharing keys.
Formally, relations $r_1$ and $r_2$ have a \textit{co-partition dependency} $g: A \hookrightarrow B$ on attribute lists $A,B$ if for all proof trees containing facts $f_1 \in r_1$, $f_2 \in r_2$, we have $\pi_B(f_1) = g(\pi_A(f_2))$ for some function $g$. If we partition by $B$ (the range of $g$) we also successfully partition by $A$ (the domain of $g$).


\begin{figure}[t]
    \centering
    \includegraphics[width=0.5\linewidth]{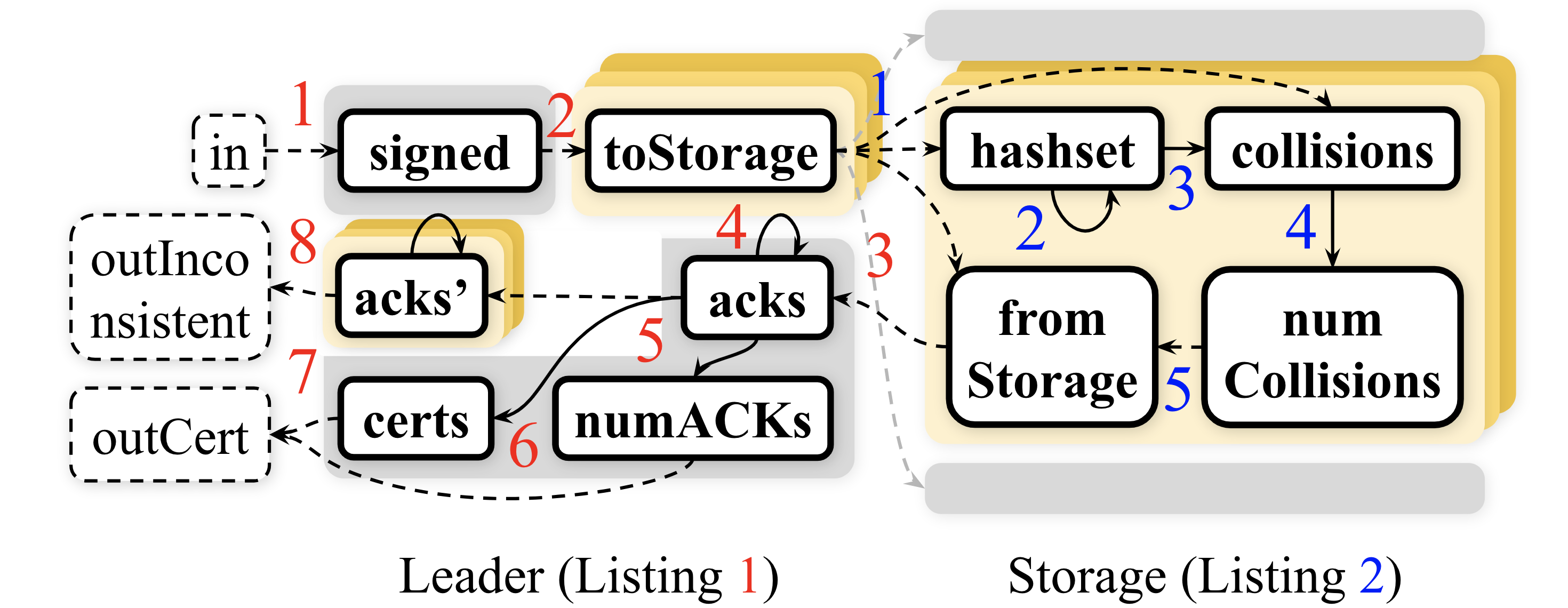}
    \caption{Running example after partitioning with dependencies.}
    \label{fig:running-example-partitioning-with-dependencies}
\end{figure}

We return to the running example to see how CDs and FDs can be combined to enable coordination-free partitioning where co-hashing forbade it.
\Cref{lst:storage} cannot be partitioned with co-hashing because \ded{toStorage} does not share keys with \ded{hashset} in \Cref{line:storage-collisions}.
No distribution policy can satisfy the co-hashing constraint if there exists two relations in the same rule that do not share keys.
However, we know that \emph{the hash is a function of the value}; there is an FD $\ded{hash.1} \to \ded{hash.2}$.
Hence partitioning on \ded{hash.2} implies partitioning on \ded{hash.1}.
The first attributes of \ded{toStorage} and \ded{hashset} are joined through the attributes of the \ded{hash} relation in all rules, forming a CD.
Let the first attributes of \ded{toStorage} and \ded{hashset}---representing a value and a hash---be $V$ and $H$ respectively: a fact $f_v$ in \ded{toStorage} can only join with a fact $f_h$ in \ded{hashset} if \ded{hash}$(\pi_{V}(f_v))$ equals $\pi_{H}(f_h)$.
This reasoning can be repeatedly applied to partition all relations by the attributes corresponding the repeated variable \ded{hashed}, as seen in \Cref{fig:running-example-partitioning-with-dependencies}.

\noindent
\textbf{Precondition:}
There exists a distribution policy $D$ for relations $r$ referenced in $C$ that partitions consistently with the CDs of $r$.

Assume we know all CDs $g$ over attribute sets $A_1,A_2$ of relations $r_1,r_2$.
A distribution policy \textbf{partitions consistently with CDs} if for any pair of facts $f_1,f_2$ over referenced relations $r_1,r_2$ in rule $\varphi$ of $C$, if $\pi_{A_1}(f_1) = g(\pi_{A_2}(f_2))$ for each attribute set, then $D(f_1) = D(f_2)$.


\tr{We describe the mechanism for systematically finding FDs and CDs in \Cref{app:partitioning-with-dependencies-checks}.}{We describe the mechanism for systematically finding FDs and CDs in the technical report.}


\noindent
\textbf{Rewrite:} Identical to Redirection with Partitioning.


\subsection{Partial partitioning}
\label{sec:partial-partitioning}
It is perhaps surprising, but sometimes additional coordination can actually help distributed protocols (like Paxos) scale.

There exist Dedalus components that cannot be partitioned even with dependency analysis.
If the non-partitionable relations are rarely written to, it may be beneficial to replicate the facts in those relations across nodes so each node holds a local copy. 
This can support multiple local reads in parallel, at the expense of occasional writes that require coordination.

We divide the component $C$ into $C_1$ and $C_2$, where relations referenced in $C_2$ can be partitioned using techniques in prior sections, but relations referenced in $C_1$ cannot.
In order to fully partition $C$, facts in relations referenced in $C_1$ must be replicated to all nodes and kept consistent so that each node can perform local processing.
To replicate those facts, inputs that modify the replicated relations are broadcasted to all nodes.

Coordination is required in order to maintain consistency between nodes with replicated facts.
Each node orders replicated inputs by buffering other inputs when replicated facts $f$ arrive, only flushing the buffer after the node is sure that all other nodes have also received $f$.
Knowledge of whether a node has received $f$ can be enforced through a distributed commit or consensus mechanism.

\noindent
\textbf{Precondition:}
$C_1$ is independent of $C_2$ and both behave like state machines.

\tr{We define ``state machines'' in \Cref{app:state-machine-decoupling} and the rewrites for partial partitioning in \Cref{app:partial-partitioning}.}{We define ``state machines'' and the rewrites for partial partitioning in the technical report.}

\section{Evaluation}
\label{sec:eval}
\changebars{}{We will refer to our approach of manually modifying distributed protocols with the mechanisms described in this paper as \emph{rule-driven rewrites}, and the traditional approach of modifying distributed protocols and proving the correctness of the optimized protocol as \emph{ad hoc rewrites}.}

In this section we address the following questions:
\begin{enumerate}
    \item How can \ourApproach{}s be applied to foundational distributed protocols, and how well do the optimized protocols scale? (\Cref{sec:optimization-and-performance})
    \item \changebars{Which of the rewrites can be applied automatically compared to manual rewrites?}{Which of the ad hoc rewrites can be reproduced via the application of (one or more) rules, and which cannot?} (\Cref{sec:compare})
    \item What is the effect of the individual \ourApproach{}s on throughput? (\Cref{sec:microbenchmark})
\end{enumerate}

\subsection{Experimental setup}
\label{sec:experimental-setup}

All protocols are implemented as Dedalus programs 
and compiled to Hydroflow~\cite{hydroflow}, a Rust dataflow runtime for distributed systems.
We deploy all protocols on
GCP using n2-standard-4 machines with 4 vCPUs, 16 GB RAM, and 10 Gbps network bandwidth, with one machine per Dedalus node.

We measure throughput/latency over one minute runs, following a 30 second warmup period.
Each client sends 16 byte commands in a closed loop.
The ping time between machines is 0.22ms.
We assume the client is outside the scope of our rewrites, and any rewrites that requires modifying the client cannot be applied.


\subsection{Rewrites and scaling}
\label{sec:optimization-and-performance}

\begin{figure*}[t]
    \centering
    \includegraphics[width=0.7\textwidth]{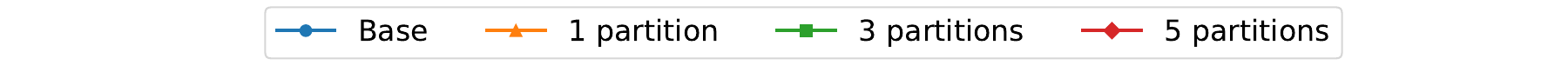}
    \centering
    \begin{subfigure}[b]{0.3\textwidth}
         \centering
         \includegraphics[width=\textwidth]{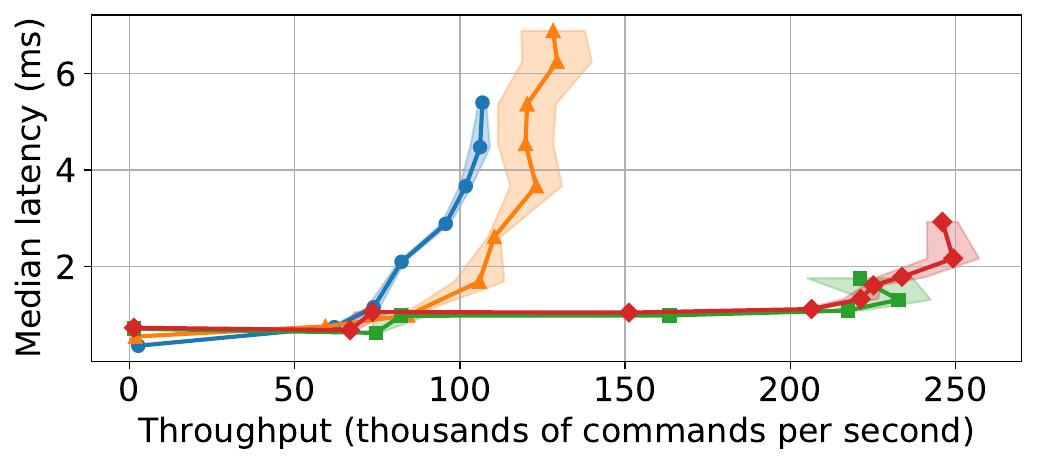}
         \caption{Voting}
         \label{fig:eval-voting-lt}
     \end{subfigure}
     \hfill
     \begin{subfigure}[b]{0.3\textwidth}
         \centering
         \includegraphics[width=\textwidth]{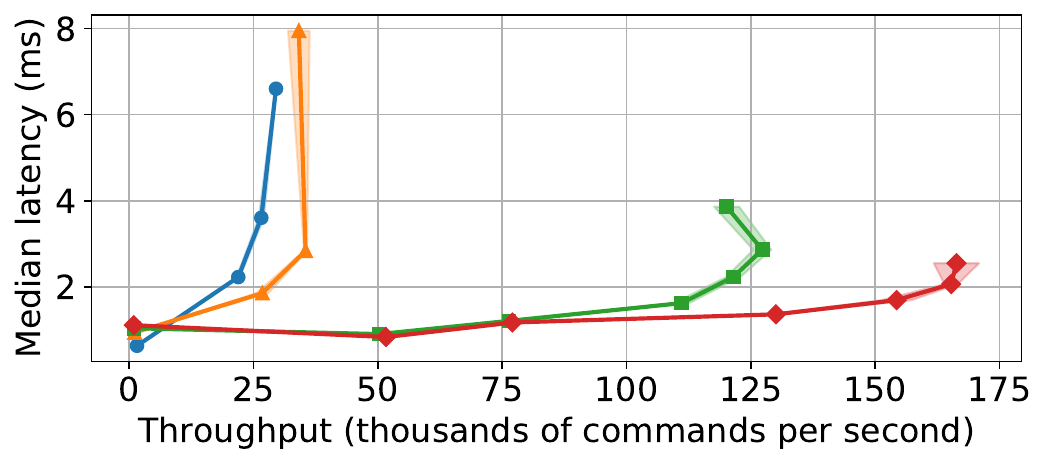}
         \caption{2PC}
         \label{fig:eval-2pc-lt}
     \end{subfigure}
     \hfill
     \begin{subfigure}[b]{0.3\textwidth}
         \centering
         \includegraphics[width=\textwidth]{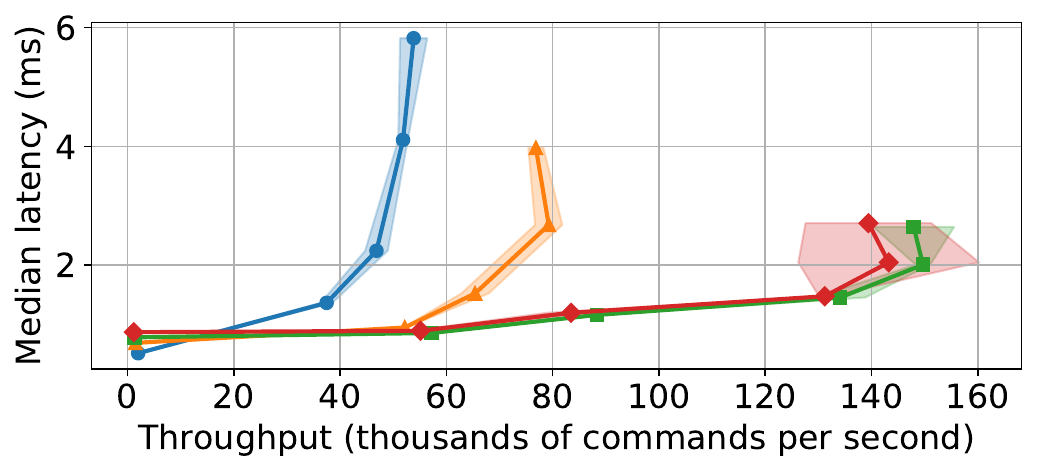}
         \caption{Paxos}
         \label{fig:eval-paxos-lt}
     \end{subfigure}
    \caption{Throughput/latency comparison between distributed protocols before and after \ourApproach{}s.}
    \label{fig:eval-lt}
\end{figure*}

We \changebars{}{manually }apply \ourApproach{}s to scale three fundamental distributed protocols---voting, 2PC, and Paxos.
We will refer to our unoptimized implementations as \DedalusVoting{}, \DedalusTwoPC{}, and \DedalusPaxos{}, and the rewritten implementations as \AutoVoting{}, \AutoTwoPC{}, and \AutoPaxos{}.
In general, we will prepend the word ``Base'' to any unoptimized implementation, ``Scalable'' to any \changebars{systematically rewritten implementation}{implementation created by applying rule-driven rewrites}, and ``\Rust{}'' to any implementation in Dedalus.
We measure the performance of each configuration with an increasing set of clients until throughput saturates, averaging across 3 runs, with standard deviations of throughput measurements shown in shaded regions.
Since the minimum configuration of Paxos (with $f=1$) requires 3 acceptors, we will also test voting and 2PC with 3 participants.

For decoupled-and-partitioned implementations, we measure scalability by changing the number of partitions for partitionable components, as seen in \Cref{fig:eval-lt}.
Decoupling contributes to the throughput differences between the unoptimized implementation and the 1-partition configuration.
Partitioning contributes to the differences between the 1, 3, and 5 partition configurations.

These experimental configurations demonstrate the scalability of the rewritten protocols.
They do not represent the most cost-effective configurations, nor the configurations that maximize throughput.
We manually applied rewrites on the critical path, selecting rewrites with low overhead, where we suspect the protocols may be bottlenecked.
Across the protocols we tested, these bottlenecks often occurred where the protocol (1) broadcasts messages, (2) collects messages, and (3) logs to disk.
\changebars{}{These bottlenecks can usually be decoupled from the original node, and because messages are often independent of one another, the decoupled nodes can then be partitioned such that each node handles a subset of messages.}
The process of identifying bottlenecks, applying suitable rewrites, and finding optimal configurations may eventually be automated.

\textbf{Voting.}
Client payloads arrive at the leader, which broadcasts payloads to the participants, collects votes from the participants, and responds to the client once all participants have voted.
Multiple rounds of voting can occur concurrently.
\DedalusVoting{} is implemented with 4 machines, 1 leader and 3 participants, achieving a maximum throughput of 100,000 commands/s, bottlenecking at the leader.

We created \AutoVoting{} from \DedalusVoting{} through \textit{Mutually Independent Decoupling, Functional Decoupling}, and \textit{Partitioning with Co-hashing}.
Broadcasters broadcast votes for the leader; they are decoupled from the leader through functional decoupling.
Collectors collect and count votes for the leader; they are decoupled from the leader through mutually independent decoupling.
The remaining ``leader'' component only relays commands to broadcasters.
All components except the leader are partitioned with co-hashing. 
The leader cannot be partitioned since that would require modifying the client to know how to reach one of many leader partitions.
With 1 leader, 5 broadcasters, 5 partitions for each of the 3 participants, and 5 collectors, the maximum configuration for \AutoVoting{} totals 26 machines, achieving a maximum throughput of 250,000 commands/s---a $2\times$ improvement over the baseline.

\textbf{2PC (with Presumed Abort).}
The coordinator receives client payloads and broadcasts \ded{voteReq} to participants.
Participants log and flush to disk, then reply with \ded{vote}s.
The coordinator collects \ded{vote}s, logs and flushes to disk, then broadcasts \ded{commit} to participants.
Participants log and flush to disk, then reply with \ded{ack}s.
The coordinator then logs and replies to the client.
Multiple rounds of 2PC can occur concurrently.
\DedalusTwoPC{} is implemented with 4 machines, 1 coordinator and 3 participants, achieving a maximum throughput of 30,000 commands/s, bottlenecking at the coordinator.

We created \AutoTwoPC{} from \DedalusTwoPC{} similarly through \textit{Mutually Independent Decoupling, Functional Decoupling}, and \textit{Partitioning with Co-hashing}.
Vote Requesters are functionally decoupled from coordinators: they broadcast \ded{voteReq} to participants.
Committers and Enders are decoupled from coordinators through mutually independent decoupling.
Committers collect \ded{vote}s, log and flush commits, then broadcast \ded{commit} to participants.
Enders collect \ded{ack}s, log, and respond to the client.
The remaining ``coordinator'' component relays commands to vote requesters.
Each participant is mutually independently decoupled into Voters and Ackers.
Participant Voters log, flush, then send \ded{vote}s; Participant Ackers log, flush, then send \ded{ack}s.
All components (except the coordinator) can be partitioned with co-hashing.
With 1 coordinator, 5 vote requesters, 5 ackers and 5 voters for each of the 3 participant, 5 committers, and 5 enders, the maximum configuration of \AutoTwoPC{} totals 46 machines, achieving a maximum throughput of 160,000 commands/s---a $5\times$ improvement.


\textbf{Paxos.}
Paxos solves consensus while tolerating up to $f$ failures.
Paxos consists of $f+1$ proposers and $2f+1$ acceptors.
Each proposer has a unique, dynamic ballot number; the proposer with the highest ballot number is the leader.
The leader receives client payloads, assigns each payload a sequence number, and broadcasts a \ded{p2a} message containing the payload, sequence number, and its ballot to the acceptors.
Each acceptor stores the highest ballot it has received and rejects or accepts payloads into its log based on whether its local ballot is less than or equal to the leader's.
The acceptor then replies to the leader via a \ded{p2b} message that includes the acceptor's highest ballot.
If this ballot is higher than the leader's ballot, the leader is preempted.
Otherwise, the acceptor has accepted the payload, and when $f+1$ acceptors accept, the payload is committed.
The leader relays committed payloads to the replicas, which execute the payload command and notify the clients.
\DedalusPaxos{} is implemented with 8 machines---2 proposers, 3 acceptors, and 3 replicas (matching \ScalaPaxos{} in \Cref{sec:compare})---tolerating $f=1$ failures, achieving a maximum throughput of 50,000 commands/s, bottlenecking at the proposer.

We created \AutoPaxos{} from \DedalusPaxos{} through \textit{Mutually Independent Decoupling, (Asymmetric)\footnote{Asymmetric decoupling is defined in \tr{\Cref{app:asymmetric-decoupling}}{the technical report}.
It applies when we decouple $C$ into $C_1$ and $C_2$, where $C_2$ is monotonic, but $C_2$ is independent of $C_1$.} Monotonic Decoupling, Functional Decoupling, Partitioning with Co-hashing}, and \textit{Partial Partitioning with Sealing}\footnote{Partitioning with sealing is defined in \tr{\Cref{app:partitioning-sealing}}{the technical report}.
It applies when a partitioned component originally sent a batched set of messages that must be recombined across partitions after partitioning.}.
P2a proxy leaders are functionally decoupled from proposers and broadcast \ded{p2a} messages.
P2b proxy leaders collect \ded{p2b} messages and broadcast committed payloads to the replicas; they are created through asymmetric monotonic decoupling, since the collection of \ded{p2b} messages is monotonic but proposers must be notified when the messages contain a higher ballot.
Both can be partitioned on sequence numbers with co-hashing.
Acceptors are partially partitioned with sealing on sequence numbers, replicating the highest ballot across partitions, necessitating the creation of a coordinator for each acceptor.
With 2 proposers, 3 p2a proxy leaders and 3 p2b proxy leaders for each of the 2 proposers, 1 coordinator and 3 partitions for each of the 3 acceptors, and 3 replicas, totalling 29 machines, \AutoPaxos{} achieves a maximum throughput of 150,000 commands/s---a $3\times$ improvement, bottlenecking at the proposer.

Across the protocols, the additional latency overhead from decoupling is negligible.

Together, these experiments demonstrate that \ourApproach{}s can be applied to scale a variety of distributed protocols, and that performance wins can be found fairly easily via choosing the rules to apply manually.
\changebars{We aim to automate this process in future work.}{A natural next step is to develop cost models for our context, and integrate into a search algorithm in order to create an automatic optimizer for distributed systems.
Standard techniques may be useful here, but we also expect new challenges in modeling dynamic load and contention.
It seems likely that adaptive query optimization and learning could prove relevant here to enable autoscaling~\cite{deshpande2007adaptive,trummer2018skinnerdb}.}


\subsection{Comparison to \traditionalApproach{}s}
\label{sec:compare}

\begin{figure}[t]
    \centering
    \includegraphics[width=0.5\linewidth]{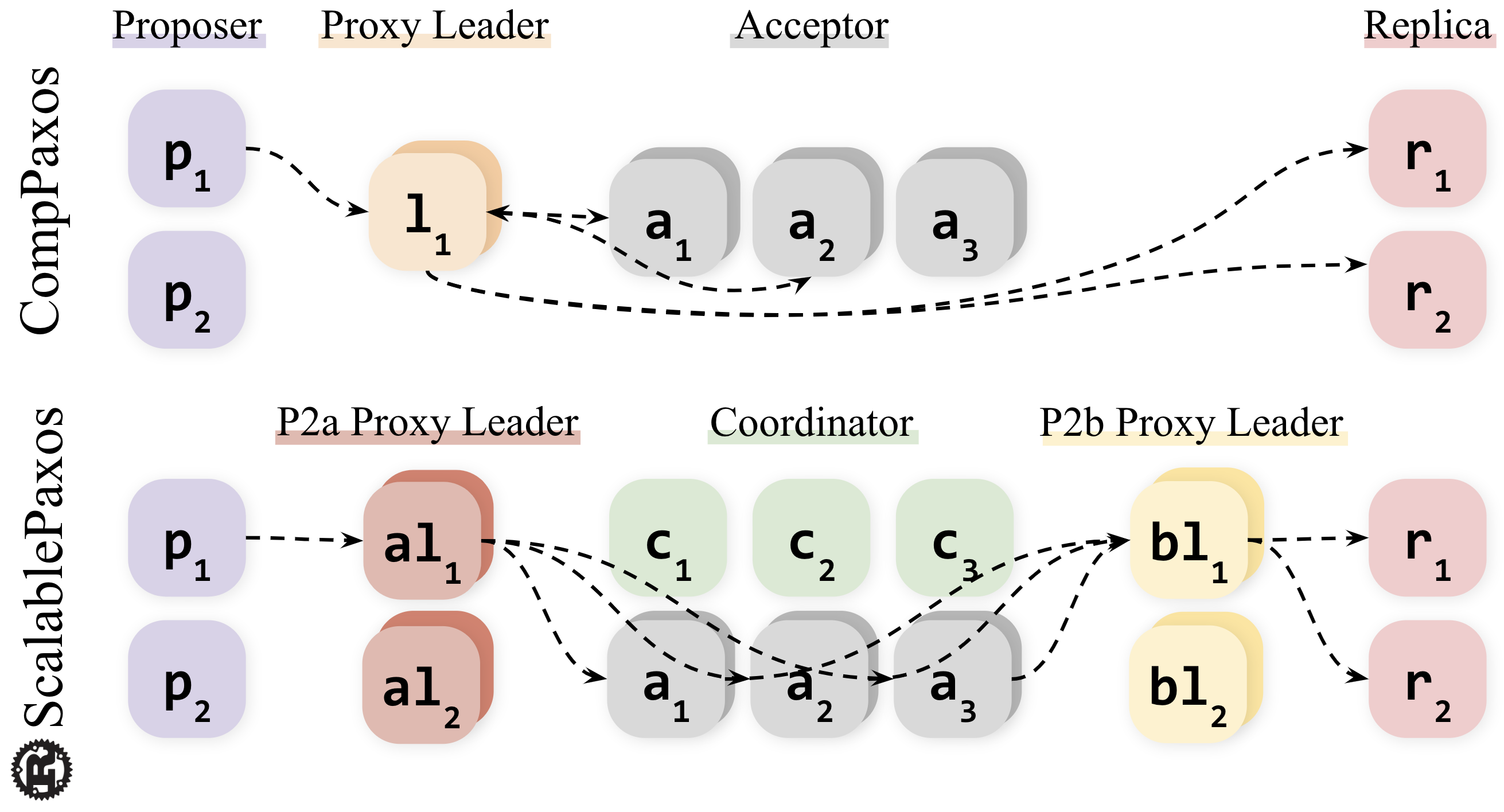}
    \caption{The common path taken by \CompPaxos{} and \AutoPaxos{}, assuming $f=1$ and any partitionable component has 2 partitions. The acceptors outlined in red represent possible quorums for leader election.}
    \label{fig:comp-vs-auto-paxos-architecture}
\end{figure}

\begin{figure}[t]
    \centering
    \includegraphics[width=0.5\linewidth]{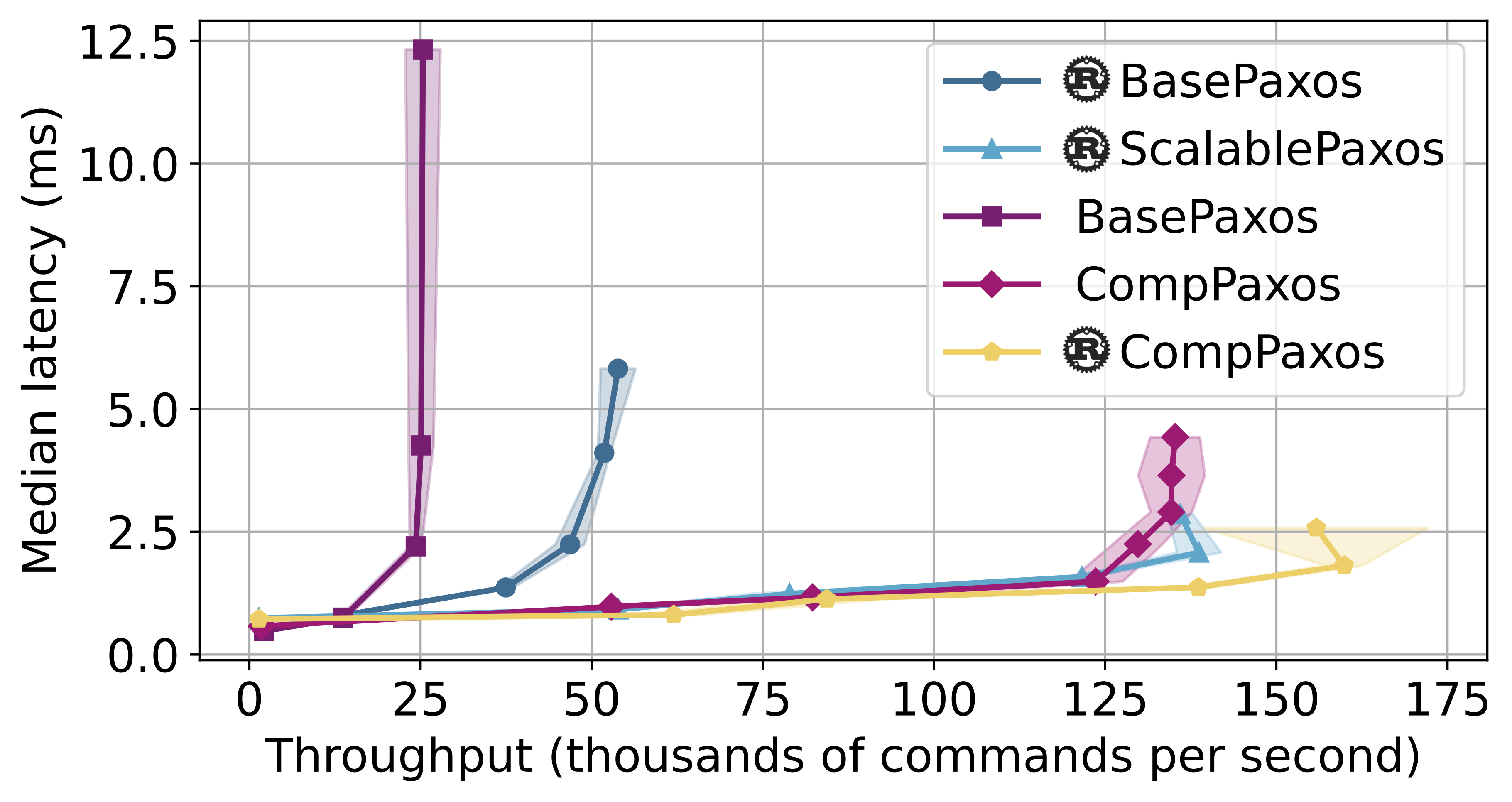}
    \caption{Throughput/latency comparison between \changebars{systematic and manual}{rule-driven and ad hoc} rewrites of Paxos.}
    \label{fig:eval-michael-comp}
\end{figure}

Our previous results show apples-to-apples comparisons between naive Dedalus implementations and \changebars{systematically rewritten Dedalus implementations}{Dedalus implementations optimized with rule-driven rewrites}.
However they do not quantify the difference between \changebars{systematically rewritten Dedalus and manually optimized code}{Dedalus implementations optimized with rule-driven rewrites and ad hoc optimized protocols} written in a more traditional procedural language.
To this effect, we compare our scalable version of Paxos to Compartmentalized Paxos~\cite{compartmentalized}.
We do this for two reasons:
(1) Paxos is notoriously hard to scale manually, and 
(2) Compartmentalized Paxos is a state-of-the-art implementation of Paxos based, among other optimizations, on manually applying decoupling and partitioning.

To best understand the merits of scalability, we choose not to batch client requests, as batching often obscures the benefits of individual scalability rewrites.

\subsubsection{Throughput comparison}
\label{sec:comp-throughput}

Whittaker et al. created Scala implementations of Paxos (\ScalaPaxos{}) and Compartmentalized Paxos (\CompPaxos{}).
Since our implementations are in Dedalus, we first compare throughputs of the Paxos implementations between the two languages to establish a baseline.
Following the nomenclature from \Cref{sec:optimization-and-performance}, implementations in Dedalus are prepended with \Rust{}, and implementations in Scala by Whittaker et al. are not.

\ScalaPaxos{} was reported to peak of 25,000 commands/s with $f=1$ and 3 replicas on AWS in 2021~\cite{compartmentalized}.
As seen \Cref{fig:eval-michael-comp}, we verified this result in GCP using the same code and experimental setup.
Our Dedalus implementation of Paxos---\DedalusPaxos{}---in contrast, peaks at a higher 50,000 commands/s with the same configuration as \ScalaPaxos{}.
We suspect this performance difference is due to the underlying implementations of \ScalaPaxos{} in Scala and \DedalusPaxos{} in Dedalus, compiled to Hydroflow atop Rust.
Indeed, our deployment of \CompPaxos{} peaked at 130,000 commands/s, and our reimplementation of Compartmentalized Paxos in Dedalus (\DedalusCompPaxos{}) peaked at a higher 160,000 commands/s, a throughput improvement comparable to the 25,000 command throughput gap between \ScalaPaxos{} and \DedalusPaxos{}.

Note that technically, \CompPaxos{} was reported to peak at 150,000 commands/s, not 130,000.
We deployed the Scala code provided by Whittaker et al. with identical hardware, network, and configuration, but could not replicate their exact result.

We now have enough context to compare the throughput between \CompPaxos{} and \AutoPaxos{}; their respective architectures are shown in \Cref{fig:comp-vs-auto-paxos-architecture}.
\CompPaxos{} achieves maximum throughput with 20 machines: 2 proposers, 10 proxy leaders, 4 acceptors (in a $2\times2$ grid), and 4 replicas.
We compare \CompPaxos{} and \AutoPaxos{} using the same number of machines, fixing the number of proposers (for fault tolerance) and replicas (which we do not decouple or partition).
Restricted to 20 machines, \AutoPaxos{} achieves the maximum throughput with 2 proposers, 2 p2a proxy leaders, 3 coordinators, 3 acceptors, 6 p2b proxy leaders, and 4 replicas.
All components are kept at minimum configuration---with only 1 partition---except for the p2b proxy leaders, which are the throughput bottleneck.
\AutoPaxos{} then scales to 130,000 commands/s, a $2.5\times$ throughput improvement over \DedalusPaxos{}.
Although \CompPaxos{} reports a $6\times$ throughput improvement over \ScalaPaxos{} from 25,000 to 150,000 commands/s in Scala, reimplemented in Dedalus, it reports a $3\times$ throughput improvement between \DedalusCompPaxos{} and \DedalusPaxos{}, similar to the $2.5\times$ throughput improvement between \AutoPaxos{} and \DedalusPaxos{}.
Therefore we conclude that the throughput improvements of \ourApproach{}s and \traditionalApproach{}s are comparable when applied to Paxos.

We emphasize that our framework
cannot realize every \changebars{manual}{ad hoc} rewrite in \CompPaxos{} (\Cref{fig:comp-vs-auto-paxos-architecture}).
We describe the differences between \CompPaxos{} and \AutoPaxos{} next.

\subsubsection{Proxy leaders}
\label{sec:compare-proxy-leaders}
\Cref{fig:comp-vs-auto-paxos-architecture} shows that \CompPaxos{} has a single component called ``proxy leader'' that serves the roles of two components in \AutoPaxos{}: p2a and p2b proxy leaders.
Unlike p2a and p2b proxy leaders, proxy leaders in \CompPaxos{} can be shared across proposers.
Since only 1 proposer will be the leader at any time, \CompPaxos{} ensures that work is evenly distributed across proxy leaders.
Our rewrites focus on scaling out and do not consider sharing physical resources between logical components.
Moreover, there is an additional optimization in the proxy leader of \CompPaxos{}.
\CompPaxos{} avoids relaying \ded{p2b}s from proxy leaders to proposers by introducing \ded{nack} messages from acceptors that are sent instead.
This optimization is neither decoupling nor partitioning and hence is not included in \AutoPaxos{}.

\subsubsection{Acceptors}
\label{sec:compare-acceptors}

\CompPaxos{} partitions acceptors without introducing coordination, allowing each partition to hold an independent ballot.
In contrast, \AutoPaxos{} can only partially partition acceptors and must introduce coordinators to synchronize ballots between partitions, because our formalism states that the partitions' ballots together must correspond to the original acceptor's ballot.
Crucially, \CompPaxos{} allows the highest ballot held at each partition to diverge while \AutoPaxos{} does not, because this divergence can introduce non-linearizable executions that remain safe for Paxos, but are too specific to generalize.
We elaborate more on this execution in \tr{\Cref{app:non-linearizable-acceptors}}{the technical report}.

Despite its additional overhead, \AutoPaxos{} does not suffer from increased latency because the overhead is not on the critical path.
Assuming a stable leader, p2b proxy leaders do not need to forward \ded{p2b}s to proposers, and acceptors do not need to coordinate between partitions.

\subsubsection{Additional differences}
\label{sec:additional-differences}

\CompPaxos{} additionally includes classical Paxos optimizations such as batching, thriftiness~\cite{epaxos}, and flexible quorums~\cite{flexiblePaxos}, which are outside the scope of this paper as they are not instances of decoupling or partitioning.
These optimizations, combined with the more efficient use of proxy leaders, explain the remaining throughput difference between \DedalusCompPaxos{} and \AutoPaxos{}.




\subsection{On the Benefit of Individual Rewrites}
\label{sec:microbenchmark}

\begin{figure}[t]
    \centering
    \includegraphics[width=0.5\linewidth]{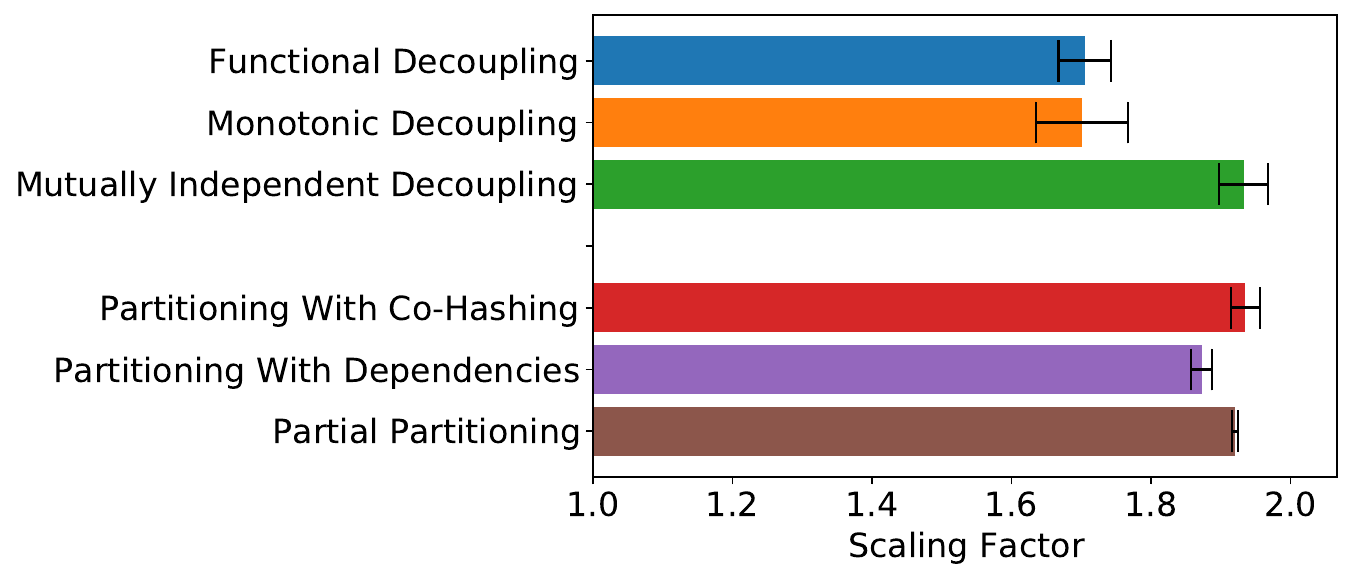}
    \caption{The scalability gains provided by each rewrite, in isolation.}
    \label{fig:eval-microbenchmark}
\end{figure}

In \Cref{fig:eval-microbenchmark}, we examine each rewrite's scaling potential. To create a consistent throughput bottleneck, we introduce extra computation via multiple AES encryptions.
When decoupling, the program must always decrypt the message from the client and encrypt its output.
When partitioning, the program must always encrypt its output.
When decoupling, we always separate one node into two. When partitioning, we always create two partitions out of one.
Thus maximum scale factor of each rewrite is $2\times$.
To determine the scaling factors, we increased the number of clients by increments of two for decoupling and three for partitioning, stopping when we reached saturation for each protocol.

Briefly, we study each of the individual rewrites using the following artificial protocols:
\begin{itemize}
    \item \textit{Mutually Independent Decoupling}: A replicated set where the leader decrypts a client request, broadcasts payloads to replicas, collects acknowledgements, and replies to the client (encrypting the response), similar to the voting protocol.
We denote this base protocol as R-set.
We decouple the broadcast and collection rules.

    \item  \textit{Monotonic Decoupling}: An R-set where the leader also keeps track of a ballot that is potentially updated by each client message.
The leader attaches the value of the ballot at the time each client request is received to the matching response.

    \item \textit{Functional Decoupling}: The same R-set protocol, but with zero replicas.
    The leader attaches the highest ballot it has seen so far to each response. It still decrypts client requests and encrypts replies as before.
    
    \item \textit{Partitioning With Co-Hashing}: A R-set.
    
    \item \textit{Partitioning With Dependencies}: A R-set where each replica records the number of hash collisions, similar to our running example.
    
    \item \textit{Partial Partitioning}: A R-set where the leader and replicas each track an integer.
    The leader's integer is periodically incremented and sent to the replicas, similar to Paxos.
    The replicas attach their latest integers to each response.
\end{itemize}

The impact on throughput varies between rewrites due to both the overhead introduced and the underlying protocol.
Note that of our 6 experiments, the first two are the only ones that add a network hop to the critical path of the protocol and rely on pipelined parallelism. The combination of networking overhead and the potential for imperfect pipelined parallelism likely explain why they achieve only about $1.7\times$ performance improvement.
In contrast, the speedups for mutually independent decoupling and the different variants of partitioning are closer to the expected $2\times$.
Nevertheless, each rewrite improves throughput in isolation as shown in \Cref{fig:eval-microbenchmark}.
\section{Related work}
\label{sec:related-work}
Our results build on rich traditions in distributed protocol design and parallel query processing. 
The intent of this paper was not to innovate in either of those domains per se, but rather to 
take parallel query processing ideas and use them to discover and evaluate rewrites for distributed protocols.

\subsection{Manual Protocol Optimizations}
There are many clever, manually-optimized variants of distributed protocols that scale by avoiding coordination, e.g.~\cite{scalog,compartmentalized,WPaxos,proteus,kauri,scalingBFT}.
These works rely on intricate modifications to underlying protocols like consensus, with manual
(and not infrequently buggy~\cite{protocolBugsList}) end-to-end proofs of correctness for the optimized protocol.
In contrast, this paper 
introduces a \changebars{systematic approach to optimization that is generalizable and}{rule-driven approach to optimization that is} correct by construction,
with proofs narrowly focused on small rewrites. 

We view our work here as orthogonal to most \changebars{manual}{ad hoc} optimizations of protocols.
Our rewrites are general and can be applied correctly to results of the \changebars{manual}{ad hoc} optimization. In future work it would be interesting 
to see when and how the more esoteric protocols cited above might benefit from further optimization using the techniques
in this paper.

Our work was initially inspired by the manually-derived Compartmentalized Paxos~\cite{compartmentalized}, from which we
borrowed our focus on decoupling and partitioning.
Our work does not \changebars{automatically}{} achieve all the optimizations of 
Compartmentalized Paxos (\Cref{sec:compare}), but it achieves the most important ones, and our results are comparable in performance.

There is a long-standing research tradition of identifying commonalities between distributed protocols that provide the same abstraction~\cite{bft700,sok,viveLaDiff,paxosVsRaft,portingPaxosOpts,bftChemistry,bftBedrock,compartmentalizedTechReport,delos}.
In principle, optimizations that apply to one protocol can be transferred to another, but this requires careful scrutiny to determine
if the protocols fit within some common framework.
We attack this problem \changebars{from a compiler perspective.}{by borrowing from the field of programming languages.}
\changebars{}{The language }Dedalus is our ``framework''; any distributed protocol expressed in Dedalus can benefit from our rewrites \changebars{automatically, without any expert oversight. Like any compiler, our work does not uncover every possible optimization a programmer can envision, but it can be quite effective in applying simple-yet-powerful rewrites, sometimes in surprisingly effective ways.}{via a mechanical application of the rules. Although our general rewrites cannot cover every possible optimization a programmer can envision, they can be applied effectively.}

\subsection{Parallel Query Processing and Dataflow}
A key intuition of our work is to rewrite protocols using techniques from distributed (``shared-nothing'') parallel databases. The core ideas go back to systems like Gamma~\cite{gamma} and GRACE~\cite{grace} in the 1980s, for both long-running ``data warehouse'' queries and transaction processing workloads~\cite{parallelDB}. Our work on partitioning (Section~\ref{sec:partitioning}) adapts ideas from parallel SQL optimizers, notably work on auto-partitioning with functional dependencies, e.g.~\cite{zhou2010incorporating}. Traditional SQL research focuses on a single query at a time. To our knowledge the literature does not include the kind of decoupling we introduce in Section~\ref{sec:decoupling}. 

Big Data systems~(e.g., \cite{mapReduce,spark,isard2007dryad}) extended the parallel query literature by adding coordination barriers and other mechanisms for mid-job fault tolerance. By contrast, our goal here is on modest amounts of data with very tight latency constraints. Moreover, fault tolerance is typically implicit in the protocols we target. As such we look for coordination-freeness wherever we can, and avoid introducing additional overheads common in Big Data systems.

There is a small body of work on parallel stream query optimization. An annotated bibliography appears in~\cite{Hirzel2018}.
Widely-deployed systems like Apache Flink~\cite{carbone2015apache} and Spark Streaming~\cite{zaharia2013discretized} offer minimal insight into query optimization.


Parallel Datalog goes back to the early 1990s (e.g.~\cite{ganguly1990framework}). A recent survey covers the state of the art in modern Datalog engines~\cite{ketsman2022modern}, including dedicated parallel Datalog systems and Datalog implementations over Big Data engines. The partitioning strategies we use in Section~\ref{sec:partitioning} are discussed in the survey; a deeper treatment can be found in the literature cited in Section~\ref{sec:partitioning}~\cite{distConstraints,distSchemes,parallelCorrectness, parallelTransfer}.

\subsection{DSLs for Distributed Systems}
We chose the Dedalus temporal logic language because it was both amenable to our optimization goals and 
we knew we could compile it to high-performance machine code via Hydroflow. Temporal logics have also been used for \emph{verification} of protocols---most notably Lamport's TLA+ language~\cite{lamport2002specifying}, which has been adopted in applied settings~\cite{newcombe2015amazon}. TLA+ did not suit our needs for a number of reasons. Most notably,
efficient code generation is not a goal of the TLA+ toolchain. 
Second, an optimizer needs lightweight checks for properties (FDs, monotonicity) in the inner loop of optimization; TLA+ is ill-suited to that case.
Finally, TLA+ was designed as a \emph{finite model checker}: it provides evidence of correctness (up to $k$ steps of execution) but no proofs. There are efforts to build symbolic checkers for TLA+~\cite{konnov2019tla+}, but again these do not seem well-suited to our lightweight setting.

Declarative languages like Dedalus have been used extensively in networking. Loo, et al. surveyed work as of 2009 including the Datalog variants NDlog and Overlog~\cite{loo2009declarative}. As networking DSLs, these languages take a relaxed ``soft state'' view of topics like persistence and consistency. Dedalus and Bloom~\cite{bloom, blooml} were developed with the express goal of formally addressing persistence and consistency in ways that we rely upon here. More recent languages for software-defined networks (SDNs) include 
NetKAT~\cite{anderson2014netkat} and P4~\cite{bosshart2014p4}, but these focus on centralized SDN controllers, not distributed systems.

Further afield, DAG-based dataflow programming is explored in parallel computing (e.g., ~\cite{blumofe1995cilk,bosilca2012dague}). While that work is not directly relevant to the transformations we study here, their efforts to schedule DAGs in parallel environments may inform future work.

\section{Conclusion}
\label{sec:conclusion}

This is the first paper to present \changebars{systematic}{general} scaling optimizations that can be safely applied to any distributed protocol, taking inspiration from traditional SQL query optimizers.
This opens the door to the creation of automatic optimizers for distributed protocols.

Our work builds on the ideas of Compartmentalized Paxos~\cite{compartmentalized}, which ``unpacks'' atomic components to increase throughput. In addition to our work on generalizing decoupling and partitioning via automation, there are additional interesting follow-on questions that we have not addressed here.
The first challenge follows from the separation of an atomic component into multiple smaller components: when one of the smaller components fails, others may continue responding to client requests.
While this is not a concern for protocols that assume omission failures, additional checks and/or rewriting may be necessary to extend our work to weaker failure models.
The second challenge is the potential liveness issues introduced by the additional latency from our rewrites and our assumption of an asynchronous network.
Protocols that calibrate timeouts assuming a partially synchronous network with some maximum message delay may need their timeouts recalibrated. This can likely be addressed in practice using typical pragmatic calibration techniques.

\section*{Acknowledgements}
This work was supported by gifts from AMD, Anyscale, Google, IBM, Intel, Microsoft, Mohamed Bin Zayed University of Artificial Intelligence, Samsung SDS, Uber, and VMware.

\bibliographystyle{ACM-Reference-Format}
{\footnotesize
\bibliography{refs}
}

\tr{
\appendix
\section{Decoupling}
\label{app:decoupling}

We will require the following terms in addition to the terms introduced in \Cref{sec:background}.

An \textbf{instance} $I$ over program $P$ is a set of facts for relations in $P$.
An \textbf{immediate consequence} operator evaluates rules to produce new facts from known facts.
$T_\varphi(I)$ over instance $I$ and rule $\varphi$ is a set of facts $f_h$ \ded{:-} $f_1, \ldots, f_n$ that is an instantiation of $\varphi$, where each $f_i$ is in $I$.
For the remainder of this paper, when we refer to instance, we mean an instance created by evaluating a sequence of immediate consequences over some set of EDB and input facts.
An instance is the state of the Dedalus program as a result of repeated rule evaluation.

A relation $r'$ is in the proof tree of $r$ if there exists facts $f' \in r'$ and $f \in r$ such that $f'$ is in the proof tree of $f$.

We assume that there is no \textit{entanglement}~\cite{dedalus}; for any fact, values representing location and time only appear in the location and time attributes respectively.
This allows us to modify the location and times of facts without worrying about changing the values in other attributes.
Our transformations may introduce entanglement when necessary.

\subsection{Mutually independent decoupling}
\label{app:mutually-independent-decoupling}




\subsubsection{Proof}
\label{app:mutually-independent-decoupling-proof}

Let $I$ be an instance over $C$ and $I'$ be an instance over $C_1$ and $C_2$.

\david{Induction hypothesis on immediate consequences.}
We will demonstrate that for each $I'$, there exists an $I$ such that (1) for relations $r$ referenced in $C_1$ and output relations of $C_1$ and $C_2$, $I$ contains fact $f$ in $r$ if and only if $I'$ contains $f$, and (2) for relations $r$ referenced in $C_2$, $I$ contains $f$ in $r$ if and only if $I'$ contains $f'$, where $f$ and $f'$ share the same values except $\pi_L(f') =$ \ded{addr2} and $\pi_L(f') =$ \ded{addr}.
By showing that $I'$ implies the existence of $I$ with the same input and output facts, each history $H'$ constructed from $I'$ must then be equivalent to some history $H$ constructed by $I$, completing our proof.

\david{Base case.}
We prove by induction; we assume that our proof is correct for any instance $I'$ constructed through $n$ immediate consequences, and show that it holds for $n+1$.
The base case of $0$ immediate consequences is trivial; $I$ and $I'$ start with the same set of EDB facts.

Let $T_\varphi(I')$ be the $n+1$-th immediate consequence over $I'$.

\david{Unchanged facts still join.}
Consider $\varphi$ where $\varphi$ is a rule of $C_1$.
Relations $r$ in $body(\varphi)$ are by definition referenced in $C_1$, so by the induction hypothesis, the same facts exist in $r$ for both $I$ and $I'$.
Therefore the same immediate consequence is possible over both $I$ and $I'$, producing the same fact at the head of $\varphi$, proving the inductive case.

\david{Facts with new location still join (can be reused).}
Now consider $\varphi$ where $\varphi$ is a rule of $C_2$.
For each $r$ in the body of $\varphi$, $I$ and $I'$ share the same facts except those facts $f$ in $I$ have $\pi_L(f) =$ \ded{addr} and $f'$ in $I'$ have $\pi_L(f') =$ \ded{addr2}.
The same immediate consequence is possible in both $I$ and $I'$ differing only in the location of facts, assuming the location attribute is only used for joins (as enforced by Dedalus).
If $\varphi$ is synchronous, then the fact at the head of $\varphi$ retains its location value in both $I$ and $I'$, proving the inductive case.
If $\varphi$ is asynchronous, then the head of $\varphi$ must be an output relation with the same location value in both $I$ and $I'$.

\david{Intuition of mechanism.}
Relations $r$ with facts generated in $C_1$ and used in $C_2$ must be forwarded from \ded{addr} to \ded{addr2} through an asynchronous message channel.
In Dedalus, this can be achieved by modifying synchronous rules into asynchronous rules, turning $r$ into an output relation of $C_1$ and an input relation of $C_2$.

\subsection{Monotonic decoupling}
\label{app:monotonic-decoupling}

\subsubsection{Checks for monotonicity}
\label{app:monotonic-decoupling-checks}

The strictest check for monotonicity requires that (1) all input relations of $C$ are persisted, and that (2) no rules of $C$ contain aggregation or negation constructs.
Fundamentally, these preconditions imply that the existence of any fact $f$ guarantees the fact $f'$ will exist at $\pi_T(f') = \pi_T(f)+1$ and $f'$ otherwise equals $f$.
If all facts $f$ over a relation $r$ has this property, then we say $r$ is \textbf{logically persisted}.
Any persisted relation (with a persistence rule) must be logically persisted, but a logically persisted relation is not necessarily persisted.
We relax the preconditions for logical persistence below.

\david{Relations don't need to be persisted if all their parents are persisted.}
A relation $r$ is logically persisted if all the relations $r$ is dependent on are logically persisted.
Consider the proof tree of any fact $f$ in $r$, and all children facts $f_c$ in $r_c$.
If all $r_c$ are logically persisted, then $f_c'$ must exist where $\pi_T(f_c') = \pi_T(f_c)+1$ and $f_c'$ otherwise equals $f_c$.
Since the rules of $C$ are monotonic (no aggregation or negation), then those facts alone are enough to guarantee the existence of $f'$ where $\pi_T(f') = \pi_T(f)+1$ and $f'$ otherwise equals $f$.
Therefore $r$ must be logically persisted.

\david{Monotonic thresholds for aggregations}
Threshold operations over monotone lattices are also monotonic~\cite{blooml}, despite involving aggregations.

\david{Garbage-collecting annotations.}
To constrain the memory footprint of monotonic components over time, we allow facts to be garbage collected from persisted relations through user annotations.
For example, in Paxos, whether a quorum is reached is a monotonic condition.
Once a quorum is reached for a particular sequence number in Paxos, the committed value cannot change, so the votes for that quorum can be safely forgotten (no longer persisted).
To allow monotonic components to forget values, we allow the user to annotate persistence rules with garbage collecting conditions.

\subsubsection{Mechanism}
\label{app:monotonic-decoupling-mechanism}

Monotonic decoupling employs both the Redirection rewrite (\Cref{sec:mutually-independent-decoupling}) and the Decoupling rewrite (\Cref{app:functional-decoupling-mechanism}) in addition to the following rewrite to persist inputs to $C_2$:

\textbf{Monotonic Rewrite:}
For all input relations $r'$ of $C_2$:
\begin{itemize}
    \item Create a relation $r''$ with all the attributes of $r'$ and replace all references of $r'$ in $C_2$ with $r''$.
    \item Add the alias and persistence rules to $C_2$:
\begin{lstlisting}[language=Dedalus, float=false, numbers=none]
r''(...,l,t) :- r'(...,l,t)
r''(...,l,t') :- r''(...,l,t), t'=t+1
\end{lstlisting}
\end{itemize}

\subsubsection{Proof}
\label{app:monotonic-decoupling-proof}
The proof is similar to that of \Cref{app:functional-decoupling-proof} for $\varphi$ where $\varphi$ is a rule of $C_1$.
We defer to the CALM Theorem~\cite{calm} for $\varphi$ where $\varphi$ is a rule of $C_2$.

\subsection{Functional decoupling}
\label{app:functional-decoupling}

\subsubsection{Mechanism}
\label{app:functional-decoupling-mechanism}

Functional decoupling employs the Redirection rewrite (\Cref{sec:mutually-independent-decoupling}) to route data from outside $C'$ to $C_2$.
Routing data from $C_1$ to $C_2$ requires the explicit introduction of asynchrony below.

\textbf{Rewrite: Decoupling.}
Given a rule $\varphi$ in $C_1$ with a head relation $r$ referenced in $C_2$:
\begin{itemize}
    \item Create a relation $r'$ with all the attributes of $r$, and replace all references of $r$ in $C_2$ with $r'$.
    \item Add the forwarding rule to $C_1$:
\begin{lstlisting}[language=Dedalus, float=false, numbers=none]
r'(...,l',t') :- r(...,l,t), forward(v,l'), delay((...,l,t,l'),t')
\end{lstlisting}
\end{itemize}

\subsubsection{Proof}
\label{app:functional-decoupling-proof}


Formally, we will demonstrate that for each instance $I'$ over $C_1$ and $C_2$, there exists an instance $I$ over $C$ such that (1) for relations $r$ referenced in $C_1$ and output relations of $C_1$ and $C_2$ (excluding input relations of $C_2$), $I$ contains fact $f$ in $r$ if and only if $I'$ contains $f$, (2) for relations $r$ referenced in $C_2$, $I$ contains fact $f$ in $r$ if and only if $I'$ contains $f'$, where $f$ and $f'$ share the same values except $\pi_L(f) =$ \ded{addr} while $\pi_L(f') =$ \ded{addr2}, and $\pi_T(f) \le \pi_T(f')$.

We again prove by induction, assuming the proof is correct for $I'$ through $n$ immediate consequences, and show that it holds for $n+1$.
The base case is trivial.

Let $T_\varphi(I')$ be the $n+1$-th immediate consequence over $I'$.

\david{Rule of $C_1$.}
Consider $\varphi'$ where $\varphi'$ is a rule of $C_1$.
If $\varphi'$ is unchanged from $\varphi$ in $C$, then the inductive hypothesis implies the same facts in both $I$ and $I'$ in all relations in the body of $\varphi'$, so $T_\varphi(I')$ implies $T_\varphi(I)$.

\david{Newly asynchronous rule (can be reused later).}
If $\varphi'$ is is a newly asynchronous rule, then let $r$ be the head of $\varphi'$ and $f'$ be the fact of $r$ in $T_{\varphi'}(I')$.
Let $\varphi$ be the original, synchronous rule.
$r$ must be an input relation of $C_2$, so we must show that $f'$ is at the head of $T_{\varphi'}(I')$ if and only if $f$ is at the head of $T_\varphi(I)$, where $\pi_L(f) =$ \ded{addr}, $\pi_L(f') =$ \ded{addr2}, and $\pi_T(f') > \pi_T(f)$.
Let $t$ be the time and $l$ be the location of all body facts in $T_{\varphi'}(I')$.
We know $t$ and $l$ are also the time and location of all body facts in $T_\varphi(I)$ by the inductive hypothesis.
$\varphi'$ differs from $\varphi$ with the two additional relations \ded{forward} and \ded{delay} added to the body.
\ded{forward} assigns $f'$ the location \ded{addr2}, while \ded{delay} sets the time of $f'$ to some non-deterministic value greater than $t$.
Since the original rule $\varphi$ was synchronous, $f$ shares the same location \ded{addr} as all other facts in $T_\varphi(I)$, and either time $t$ (if $\varphi$ is deductive) or $t+1$ (if $\varphi$ is inductive).
This proves the inductive hypothesis: $f$ and $f'$ share the same values except $\pi_L(f) =$ \ded{addr} while $\pi_L(f') =$ \ded{addr2}, and $\pi_T(f) \le \pi_T(f')$.

\david{Rule of $C_2$.}
Now consider $\varphi'$ where $\varphi'$ is a rule of $C_2$.
$\varphi'$ in $C_2$ is unchanged from $\varphi$ in $C$.
Since we assumed that $C_2$ is functional, $\varphi'$ contains at most one IDB relation in its body.
If there are only EDB relations in its body, then the facts in those EDBs are the same in both $I'$ and $I$, completing the proof.
If there is one IDB relation $r_b$ in its body, then by the induction hypothesis, the fact $f_b'$ of $r_b$ in $I'$ implies fact $f_b$ in $I$, where $\pi_L(f_b) =$ \ded{addr} and $\pi_L(f_b') =$ \ded{addr2}, $\pi_T(f_b) < \pi_T(f_b')$, and $f_b'$ otherwise equals $f_b$.
All remaining relations in the body of $\varphi$ must be EDBs with the same facts across all locations and times.
Therefore, any immediate consequence $T_{\varphi'}(I')$ with $f_b'$ in its body and $f'$ in its head implies $T_\varphi(I)$ with $f_b$ in its body and $f$ in its head.
If $\varphi$ is synchronous, then $f'$ and $f_b'$ share the same location and time, $f$ and $f_b$ share the same location and time, and $f'$ and $f$ are otherwise equal, completing the proof.

\david{Outputs of $C_2$ match outputs of $C$.}
If $\varphi$ is asynchronous, then $f'$ and $f$ are output facts.
Then $f'$ and $f$ share the same location (the destination), whereas for time, the facts only need to satisfy the inequalities $\pi_T(f') > \pi_T(f_b')$ and $\pi_T(f) > \pi_T(f_b)$.
In other words, the range of possible values for $\pi_T(f')$ is $(\pi_T(f_b'), \infty)$ and the range of possible values for $\pi_T(f)$ is $(\pi_T(f_b), \infty)$.
Since $\pi_T(f_b') > \pi_T(f_b)$, the range $(\pi_T(f_b'), \infty)$ must be a sub-range of $(\pi_T(f_b), \infty)$.
Therefore, given $f'$ in $T_{\varphi'}(I')$, an immediate consequence $T_\varphi(I)$ with $f$ is always possible where $\pi_T(f) = \pi_T(f')$ and $f = f'$, completing the proof.

\subsection{State machine decoupling}
\label{app:state-machine-decoupling}

Although this decoupling technique has since been cut from the paper, we still include it since its preconditions and proofs are referenced in \Cref{app:partial-partitioning}.

In a state machine component, any pair of facts that are combined (say via join or aggregation) at time $t$ must be the result of
inputs at time $t$ and the order of inputs prior, but the exact value of $t$ is irrelevant.
In these cases, we want to guarantee that (a) facts that co-occur at time $t$ in $C$ 
will also be processed together at some time $t'$ in $C_2$, and (b) the inputs of $C$ prior to $t$ match the inputs of $C_2$ prior to $t'$.
To meet this guarantee, we collect facts from $C_1$ to $C_2$ into sequenced batches.

\textbf{Precondition:}
$C_1$ is independent of $C_2$, and $C_2$ behaves like a \textbf{state machine}.

Before we formalize what it means to behave like a state machine, a couple of definitions are helpful.

\begin{definition}[Existence dependency]
\david{Do folks understand this notation?}
Relation $r$ has an \textit{existence dependency} on input relations $\overline{r_{in}}$ if $r$ is empty whenever there is no input;
that is, $r = \emptyset$ in any timestep when $\bigwedge_{r_{i} \in \overline{r_{in}}} r_{i} = \emptyset$.
\end{definition}\heidi{Nit pick: Sometimes terms in definitions are textbf and sometimes textit}
\nc{Should this not be Union rather than AND?}

\begin{definition}[No-change dependency]
Relation $r$ has a \textit{no-change dependency} on input relations $\overline{r_{in}}$ if 
$r$'s contents remain unchanged in a timestep when the inputs are empty.
That is, if $\bigwedge_{r_{i} \in \overline{r_{in}}} r_{i} = \emptyset$ at
timestep $t$, then $r$ contains exactly the same facts at timestep $t$ as it did in timestep $t-1$.
\end{definition}


Formally, $C_2$ is a state machine if (a) all referenced relations have either existence or no-change dependencies on the inputs, and (b) outputs of $C_2$ have existence dependencies on the inputs.
Condition (a) ensures that the component is insensitive to the passing of time(steps), and (b) ensures that the passing of time(steps) does not affect output content. \jmh{I'm not sure I 
understand the argument here fully, hence my explanation is not too convincing.} 
\nc{It should definitely affect output timing but you may want to say it doesn't affect output content? I'm also not entirely sure this is true. By this definition, Paxos is not a state machine? Timing definitely affects output content it's just that the system remains linearizable. You may want to use another word than state machine?}

\subsubsection{Checks for state machine behavior}
\label{app:state-machine-decoupling-checks}

We provide conservative tests on relations to identify existence and no-change dependencies.
A relation $r$ has an existence dependency on input relations $\overline{r_{in}}$ if for all rules $\varphi$ in the proof tree of $r$, (1) $\varphi$ does not contain \ded{t'=t+1}, and (2) the body of $\varphi$ contains at least one non-negated relation $r'$ where either $r'$ is an input or $r'$ also has an existence dependency on $\overline{r_{in}}$.

A relation $r$ has a no-change dependency on input relations $\overline{r_{in}}$ if:
\begin{enumerate}
    \item \textbf{Explicit persist.} If there is an inductive rule $\varphi$ with $r = head(\varphi)$, $\varphi$ must be the persistence rule. Then $r$ is persisted.
    \item \textbf{Implicit persist.} If is no such inductive rule, then for all (non-inductive) rules $\varphi$ where $r = head(\varphi)$, the body of $\varphi$ contains only EDBs and relations $r'$ where $r'$ has a no-change dependency on $\overline{r_{in}}$.
    \item \textbf{Change only on inputs.} If there is such an inductive rule, then we also allow rules $\varphi$ where $r = head(\varphi)$ to contain at least one non-negated relation $r'$ in the body, where either $r' \in \overline{r_{in}}$ or $r'$ has an existence dependency on $\overline{r_{in}}$.
\end{enumerate}

\subsubsection{Mechanism}
\label{app:state-machine-decoupling-mechanism}

\david{Intuition of mechanism.}
To guarantee coexistence of facts, rewrites for state machine decoupling must preserve the order and batching of inputs.
Similar to the rewrites above, we create new relations and asynchronous forwarding rules for relations in $C_1$ referenced in $C_2$.
To preserve ordering and batching of inputs, we create additional rules in $C_1$ to track the number of previous batches and the current batch size (the number of output facts to $C_2$).
Then $C_2$ ensures that all previous batches have been processed and the current batch has arrived before processing any input fact in the current batch.

\david{Intuition of causality preservation.}
Unlike any decoupling techniques described so far, we cannot reroute rules from other components $C'$ to $C_2$; those input facts must be batched and ordered by $C_1$.
Intuitively, if all inputs are routed through $C_1$, then the batching and ordering on $C_1$ is a feasible batching and ordering on $C$.
$C_2$ must then process its inputs with the same batching and ordering to guarantee correctness.
Were facts $f$ to arrive at $C_2$ without batching or ordering information from $C_1$, then $f$ may happen-after some input $f'$ in $C_1$ but be processed before $f'$ is processed at $C_2$, violating causality.

Our rewrites append a new attribute $T_1$ to relations forwarded to $C_2$ from $C_1$.
This attribute represents the time on $C_1$ when each input fact existed, allowing $C_2$ to process facts in the same order and batches as $C_1$ even with non-deterministic message delay.

\textbf{Rewrite: Batching.}
Given a rule $\varphi$ in either $C_1$ or another component $C'$ whose head $r$ is referenced in $C_2$, we add the following rules to $C_1$:

\begin{lstlisting}[language=Dedalus, float=false, mathescape=true]
# Create a relation $r'$ with all the attributes of $r$, with an additional attribute $T_1$, and forward to $C_2$.
r'(...,t,l',t') :- r(...,l,t), forward(l,l'), delay((...,l,t,l'),t') |\label{line:sm-c1-forward}|
# Count the number of facts of $r$ for any time. Assumes that count evaluates to 0 if $r$ is empty.
rCount(count<...>,l,t) :- r(...,l,t) |\label{line:sm-c1-count-r}|
# Sum the size of the batch across relations $r_i$.
batchSize(n,l,t) :- r1Count(n1,l,t), r2Count(n2,l,t), ..., n=n1+n2+... |\label{line:sm-c1-batch-size}|
# Record whenever the batch size is non-zero.
batchTimes(t,l,t') :- batchSize(n,l,t), n!=0, t'=t+1 |\label{line:sm-c1-batch-times}|
batchTimes(t,l,t') :- !batchSize(n,l,t), batchTimes(t,l,t), t'=t+1 |\label{line:sm-c1-batch-times-persist}|
# Send the batch size and times to $C_2$.
inputs(n,t,prevT,l',t') :- batchTimes(prevT,l,t), batchSize(n,l,t), n!=0, forward(l,l'), delay((n,t,prevT,l,l'),t') |\label{line:sm-c1-send-batch}|
inputs(n,t,0,l',t') :- !batchTimes(prevT,l,t), batchSize(n,l,t), n!=0, forward(l,l'), delay((n,t,prevT,l,l'),t')
\end{lstlisting}

On $C_2$, we modify all references of input relations $r$ to \ded{rSealed} and add the following rules:
\begin{lstlisting}[language=Dedalus, float=false, mathescape=true]
r''(...,t1,l,t') :- r'(...,t1,l,t)
# Count the number of facts of $r$ for batch $t1$.
r''(...,t1,l,t') :- r''(...,t1,l,t), t'=t+1 |\label{line:sm-c2-r-persist}|
rCount(count<...>,t1,l,t) :- r''(...,t1,l,t) |\label{line:sm-c2-count-r}|
# Count the number of facts across $r_i$ for batch $t1$.
recvSize(n,t1,l,t) :- r1Count(n1,t1,l,t), r2Count(n2,t1,l,t), ..., n=n1+n2+... |\label{line:sm-c2-batch-size}|
# Check if this batch has been received and the previous batch has been processed.
inputs(n,t1,prevT,l,t') :- inputs(n,t1,prevT,l,t), t'=t+1 |\label{line:sm-c2-input-persist}|
canSeal(t1,l,t) :- recvSize(n,t1,l,t), inputs(n,t1,prevT,l,t), sealed(prevT,l,t) |\label{line:sm-c2-batch-can-seal}|
canSeal(t1,l,t) :- recvSize(n,t1,l,t), inputs(n,t1,0,l,t)
# Mark this batch as processed.
sealed(t1,l,t') :- canSeal(t1,l,t), t'=t+1 |\label{line:sm-c2-batch-sealed}|
sealed(t1,l,t') :- sealed(t1,l,t), t'=t+1 |\label{line:sm-c2-seal-persist}|
# Can process facts at time $t1$.
rSealed(...,l,t) :- r''(...,t1,l,t), canSeal(t1,l,t) |\label{line:sm-c2-rSealed}|
\end{lstlisting}

Note that whenever a time $t1$ is sealed on $C_2$, facts in $r''$ and \ded{inputs} can be garbage collected.
Facts in \ded{sealed} can be garbage collected if a higher timestamp has been sealed.
We omit these optimizations for simplicity.

\subsubsection{Proof}
\label{app:state-machine-decoupling-proof}

Our proof relies on $C_2$ processing inputs in the same order and batches as $C_1$.
For simplicity, we denote $I_{\overline{r},t}$ as the set of facts $f$ in $I$ where $f$ is a fact of relation $r$ in $\overline{r}$ and $\pi_T(f) = t$.

Formally, we will prove that for each instance $I'$ there exists $I$ such that (1) for relations $r$ referenced in $C_1$ and output relations of $C_1$ and $C_2$ (excluding input relations of $C_2$), $I$ contains fact $f$ in $r$ if and only if $I'$ contains $f$, and (2) for the set of relations $\overline{r'}$ referenced in $C_2$ (and $\overline{r}$ corresponding to relations with \ded{rSealed} replaced with $r$), for any time $t'$ where at least one input relation \ded{rSealed} of $C_2$ is not empty in $I'$ and contains fact $f_{in}'$, let $t = \pi_{T_1}(f_{in}')$. 
We must have $I'_{\overline{r'},t'} = I_{\overline{r},t}$ when facts in \ded{rSealed} are mapped to $r$, and location and time are ignored.

For rules $\varphi'$ of $C_1$, the inductive proof is identical to previous proofs.

\david{Proof for rules on $C_2$.}
Now consider $\varphi'$ where $\varphi'$ is a rule of $C_2$ and $\varphi'$ corresponds to $\varphi$ in $C$.
Let $t'$ be the time of immediate consequence $T_{\varphi'}(I')$ such that for all facts $f'$ of relations $r$ in the body of $\varphi'$, $\pi_T(f') = t'$.
If all input relations (\ded{rSealed}, not $r$) are empty for $t'$ in $I'$, then $I'$ cannot produce output facts at $t'$, since output relations must have existence dependencies on the input relations.

\david{Case where inputs to $C_2$ are non-empty.}
If at least one input relation \ded{rSealed} is not empty for $t'$ in $I'$, we must show $I'_{\overline{r'},t'} = I_{\overline{r},t}$.
Let $t = \pi_{T_1}(f_{in}')$ for some fact $f_{in}'$ in \ded{rSealed} with $\pi_T(f_{in}') = t'$.
Let $t_<'$ be the time of the previous input on $I'$; formally, for all input relations \ded{rSealed} of $C_2$, there is no $t''$ where $t_<' < t'' < t'$ such that $I'_{\ded{rSealed}, t''}$ is non-empty.
Similar to how we construct $t$ from $t'$, let $t_< = \pi_{T_1}(f_{in}')$ for some fact $f_{in}'$ in \ded{rSealed} with $\pi_T(f_{in}') = t_<'$.

\david{Previous sealed input on $I'$ corresponds to previous input on $I$; order is preserved.}
We first show that $t_<$ is the time of the previous input on $I$; formally, for input relations \ded{rSealed} of $C_2$, for all corresponding $r$, there is no $t''$ where $t_< < t'' < t$ such that $I_{r, t''}$ is non-empty.
We prove by contradiction, assuming such $t''$ exists.
In order for a fact $f_{in}'$ in \ded{rSealed} in $I'$ to have $\pi_T(f_{in}') = t'$, \ded{canSeal} must contain the fact \ded{canSeal}$(t,\ded{addr2},t')$, which is only possible if \ded{sealed} contains the fact \ded{sealed}$(t'',\ded{addr2},t')$ and \ded{inputs} contains the fact \ded{inputs}$(n,t,t'',\ded{addr2},t')$.
\ded{sealed}$(t'',\ded{addr2},t')$ implies the fact \ded{canSeal}$(t'',\ded{addr2},t_<')$ where $t_<' = \pi_T(f_{in}')$.
In order for some fact $f$ in $r$ to join with \ded{canSeal}$(t'',\ded{addr2},t_<')$ to create $f_{in}$, we must have $\pi_{T_1}(f) = t'' = \pi_{T_1}(f_{in})$.
By definition, $\pi_{T_1}(f_{in}) = t_<'$, so $t_<' = t''$.

\david{Since $t_<$ was the previous input on $I$ and $I'$, there are no further inputs in between and we can infer state using no-change and existence dependencies.}
Knowing that no inputs facts exist with times between $t_<$ and $t$ in $I$ and between $t_<'$ and $t'$ in $I'$, we can use the existence and no-change dependencies of relations $r$ in $I$ to reason about the instances $I'_{\overline{r'},t'}$ and $I_{\overline{r},t}$.
By the induction hypothesis, we have $I'_{\overline{r'},t_<'} = I_{\overline{r},t_<}$.
We can now reason about the instances $I'$ and $I$ at times $t'-1$ and $t-1$, respectively.
For all relations $r$ with existence dependencies on the inputs, $I'_{r,t'-1}$ and $I_{r,t-1}$ must both be empty.
For all relations $r$ with no-change dependencies on the inputs, $I'_{r,t'-1} = I'_{r,t_<'}$ and $I_{r,t-1} = I_{r,t_<}$, so $I'_{r,t'-1} = I_{r,t-1}$.

\david{Induction on immediate consequences after the input was received.}
Consider an immediate consequence $T_{\varphi'}(I')$ with facts $f'$ in the body of $\varphi'$ where $\pi_T(f') = t'$.
We find the set of facts in the proof tree of $f'$ that have no proof tree (because they are inputs or EDBs) or have parents in the tree with time $t'-1$.
Inputs and EDBs are the same between $I'$ and $I$ at time $t'$ and $t$, due to our sealing mechanism.
Since we know $I'_{r,t'-1} = I_{r,t-1}$, any parent fact in the tree with time $t'-1$ exists in $I$ with time $t-1$ and evaluates to the same fact.
Therefore, the same series of immediate consequences are possible in $I$ to produce each fact in $I'_{r,t'-1}$.
Thus $I'_{\overline{r'},t'} = I_{\overline{r},t}$.

\david{Outputs of $C_2$ match outputs of $C$.}
If $\varphi'$ is an asynchronous rule, then the head of $\varphi'$ is an output relation $r$, and we have to show that the fact $f'$ in $r'$ of immediate consequence $T_{\varphi'}(I')$ is equivalent to some fact $f$ in $r$ of immediate consequence $T_\varphi(I)$.
Output relations must have existence dependencies in $C_2$ by assumption, so given $f_b$ in the body of $T_{\varphi'}(I')$ with $\pi_T(f) = t'$, the input relations are not empty at $t'$ in $I'$, and $I'_{\overline{r'},t'} = I_{\overline{r},t}$.
$f_b'$ of $\varphi'$ for $I'$ must correspond to $f_b$ in the body of $\varphi$ for $I$, $\pi_L(f_b) =$ \ded{addr} and $\pi_T(f_b) = t$ instead of $\pi_L(f_b') =$ \ded{addr2} and $\pi_T(f_b') = t'$.
The facts $f'$ and $f$ at the head of $\varphi'$ and $\varphi$ must be the same as well, non-deterministic with the constraints $\pi_T(f') > t'$ and $\pi_T(f) > t$.
We know $t' > t$ due to the addition of an asynchronous channel, therefore $\pi_T(f) = \pi_T(f')$ is always possible, and $f = f'$, completing the proof.

\subsection{Asymmetric decoupling}
\label{app:asymmetric-decoupling}

In this section, we consider decoupling where $C_1$ and $C_2$ are mutually dependent and where $C_2$ is independent of $C_1$ instead.


If $C_1$ and $C_2$ are both monotonic, then they can be decoupled through the Redirection With Persistence rewrite (\Cref{sec:monotonic-decoupling}), according to the CALM Theorem~\cite{calm}.

Now consider $C_2$ independent of $C_1$, where $C_2$ exhibits useful properties for decoupling.
Intuitively, although $C_2$ forwards facts to $C_1$, we can treat the time in which $C_1$ processes inputs as the ``time of input arrival'' while allowing $C_2$ to process inputs first.

This presents a problem: $C_2$ might produce outputs ``too early'', violating well-formedness. \kaushik{what does well-formedness mean?}
To preserve well-formedness, we introduce a rewrite to delay output facts $f'$ derived from input fact $f$ of $C_2$ until $C_1$ acknowledges it has processed $f$. \chris{I was initially a bit confused about how these two paragraphs aren't just like rest of the section, but with the names of $C_1$ and $C_2$ swapped around. In my original read-through of the earlier parts of this section, it seemed that the only difference between $C_1$ and $C_2$ was that $C_1$ was at the "original" addr while $C_2$ was moved to a new addr2.}
\kaushik{+1 to Chris's comment}

\textbf{Precondition:} $C_2$ is independent of $C_1$, and $C_2$ is (1) a state machine and (2) either monotonic or functional.

\subsubsection{Mechanism}
\label{app:asymmetric-decoupling-mechanism}

To decouple, first apply the Batching rewrite (\Cref{app:state-machine-decoupling-mechanism}) for all rules in $C_2$ whose head $r$ is referenced in $C_1$, replacing $r$ with \ded{rSealed}.
Populate \ded{forward} with both \ded{forward(addr,addr2)} and \ded{forward(addr2,addr)}.
Then apply the forwarding rewrite for all rules in another component $C'$ whose head is referenced in $C_2$.
Finally, perform the following rewrite to delay outputs of $C_2$:

\textbf{Rewrite: Batch Acknowledgement.}
Add the following rules to track which batches $C_1$ has processed:
\begin{lstlisting}[language=Dedalus, float=false, mathescape=true]
# Component $C_1$.
batchACK(t2,l',t') :- canSeal(t2,l,t), forward(l,l'), delay(t,t')
# Component $C_2$.
batchACK(t2,l,t') :- batchACK(t2,l,t), t'=t+1
\end{lstlisting}

For each rule $\varphi$ in $C_2$ with output relation \ded{out}, create the relation \ded{outP} with two additional attributes--- $T2$ and $L'$---representing the derivation time of the fact and the destination, and add the following rules:
\begin{lstlisting}[language=Dedalus, float=false, mathescape=true]
# Replace $\varphi$ with this rule.
outP(t,l',...,l,t) :- ...
# Find the batchACK each output must wait on.
outBatchTime(max<prevT>,l,t) :- batchTimes(prevT,l,t), outP(t2,l',...,l,t), prevT <= t2
# Outputs derived at time t2 can now send.
outCanSend(t2,l,t) :- outBatchTime(t2,l,t), batchACK(t2,l,t)
# Buffer outputs until they can send.
outP(t2,l',...,l,t') :- outP(t2,l',...,l,t), t'=t+1, !outCanSend(t2,l,t)
# Send the outputs.
out(...,l',t') :- outP(t2,l',...,l,t), outCanSend(t2,l,t), delay(t,t')
\end{lstlisting}

\subsubsection{Proof}
\label{app:asymmetric-decoupling-proof}

Since input relations $r$ to $C_1$ are now buffered and replaced with \ded{rSealed}, the arrival time of input facts in the transformed component $C$ no longer corresponds to the processing time.
This poses a problem for our proof; in the original component $C$, the arrival time of input facts \emph{is} the processing time.
Note that since messages are sent over an asynchronous network, any entity that sends and receives messages from $C$ can only observe the ``send'' time $t_s$ of each input fact $f$, where $t_s < \pi_T(f)$.
Intuitively, an input fact that is sent at time $t_s$ and in-network for $t$ seconds is processed identically to a fact sent at time $t_s$, in-network for $t'$ seconds, and buffered for $b$ seconds, so long as $t' + b = t$.

To formalize this intuition, we introduce the \textbf{observable instance} $\mathcal{I}$, which contains a set of facts representing the send times of input facts and arrival times of output facts.
An instance $I$ is \textbf{observably equivalent} to $\mathcal{I}$ if for all facts $f$ in relation $r$ of $I$, 
(1) if $r$ is an input relation, then $f$ exists if and only if there exists $f_s$ in $\mathcal{I}$ where $f$ equals $f_s$ except $\pi_T(f) > \pi_T(f_s)$, and 
(2) if $r$ is an output relation, then there exists $f$ in $I$ if and only if there exists $f$ in $\mathcal{I}$.

The history of an instance $I$ is constructed using the wall-clock times of inputs and outputs of its observably equivalent instance $\mathcal{I}$.
Therefore, any instances $I_1$ and $I_2$ that are observably equivalent to the same $\mathcal{I}$ share the same histories.

\david{Proof.}
We will prove that for each instance $I'$ there exists $I$ and the observable instance $\mathcal{I}$ such that 
(1) $I$ and $I'$ are both observably equivalent to $\mathcal{I}$, and
(2) Each fact $f$ of relation $r$ referenced by $C_1$ in $I'$ exists if and only if $f$ exists in $I$.

\david{Selecting input arrival times for $I$.}
We select the time of input facts in $I$ such that its inputs are observably equivalent to $\mathcal{I}$.
For each input relation $r$ in $C$:
(1) If $r$ is an input of $C_1$, each fact $f$ of $r$ in $I'$ exists if and only if $f$ exists in $I$, and
(2) If $r$ is an input of $C_2$, given $f'$ of $r$ in $I'$ with $t' = \pi_T(f')$, let the seal time be $t_s = \pi_{T2}(f_s)$ in $f_s$ of \ded{canSeal} in $I'$; let $f$ equal $f'$ except $\pi_L(f) =$ \ded{addr} and $\pi_T(f) = \pi_T(f_s)$; $f'$ exists if and only if $f$ exists in $I$, and
(3) If no such $f_s$ exists, then let $\pi_T(f) = \pi_T(f')$.
Note that in case 2, $\pi_T(f) > \pi_T(f')$ due to asynchrony from $C_2$ to $C_1$.
Therefore, input facts $f$ in $I$ correspond to input facts $f'$ in $I'$ where $\pi_T(f) \ge \pi_T(f')$, thus the inputs of $I$ are observably equivalent to $\mathcal{I}$.

\david{Prove $C_1$.}
We first prove the claim that $I'$ and $I$ are equivalent for all facts over relations referenced in $C_1$.
We prove by induction over the immediate consequence of rules $\varphi$ whose head $r$ is referenced in $C_1$.
If $\varphi$ is an original rule of $C_1$, then by the inductive hypothesis, the claim is true.
If $\varphi$ is a rule of another component $C'$, then the rule is unchanged and the claim is still true.

\david{Must show that in $I'$, the relations in the body had state at $t'$ that matches the state of $I$ at time $t$, and $t' < t$.}
Now consider $\varphi'$ where the head of $\varphi'$ is \ded{rSealed} corresponding to some $r$ after transformation.
$\varphi'$ is a rule we introduced so there is no immediate consequence over it in $I$.
Let $\varphi$ refer to the rule with $r$ at its head instead of \ded{rSealed}; $\varphi$ is a rule of $C_2$.
We can show that the immediate consequence $T_\varphi(I')$ at time $t'$ is possible if and only if $T_\varphi(I)$ at some time $t > t'$ is possible; this is true if for all relations $r_b \in body(\varphi)$, $I'_{r_b,t'} = I_{r_b,t}$.
Then as long as $t > t'$, the facts of $r$ of immediate consequence $T_\varphi(I')$ can be asynchronously delivered from $C_2$ to $C_1$ at some time $t_d$, where $t \ge t_d > t'$, and be sealed at time $t$ in a an immediate consequence $T_{\varphi'}(I')$ over $\varphi'$, resulting in the same facts in \ded{rSealed} in $I'$ and $r$ in $I$.

\david{Must show that the inputs of $C_2$ at $I'$ up to $t'$ are the same as the inputs at $I$ up to $t$.}
Since $C_2$ behaves like a state machine, the facts of $r_b$ are dependent on the input facts of $C_2$ and their ordering.
Since $C_2$ is either functional or monotonic, the facts of $r_b$ are not dependent on the ordering of the input facts; this can be proven by treating each $r_b$ as an output relation, then reapplying proofs from \Cref{app:functional-decoupling-proof,app:monotonic-decoupling-proof}.
Let $\overline{r}$ be the input relations of $C_2$.
We can prove that facts $f'$ of $\overline{r_{in}}$ in $I'$ exists if and only if $f$ of $\overline{r_{in}}$ in $I$ exists, where $f'$ equals $f$ except $\pi_L(f') =$ \ded{addr2} while $\pi_L(f) =$ \ded{addr}, and if $\pi_T(f') = t'$ then $\pi_T(f) = t$, but if $\pi_T(f') < t'$, then $\pi_T(f) < t$.
In other words, $I$ and $I'$ share inputs at $t'$ and $t$ and all previous inputs, but previous inputs may be out-of-order.

We prove by contradiction.
First consider some $f_i'$ of an input relation of $C_2$ in $I'$ where $\pi_T(f_i') = t'$, but there is no $f_i$ in $I$ where $\pi_T(f_i) = t$.
Since $r$ an output of $C_2$ and input of $C_1$, it is batched according to the batching rewrite, and there must be a fact $f_c$ in \ded{canSeal} of $I'$ with $\pi_{T2}(f_c) = t'$ signalling when facts in $r$ can be processed.
By construction of $I$, $f_i'$ exists in $I'$ if and only if $f_i$ exists in $I$, a contradiction.
Now consider some $f_i'$ where $\pi_T(f_i') < t'$, but there is no $f_i$ where $\pi_T(f_i) < t$.
Either there is some $f_c$ as above, or $\pi_T(f_i) = \pi_T(f_i')$ by construction of $I$.
Since $t' < t$, $t' = \pi_T(f_i) < t$, completing the proof.
Therefore, for inputs and outputs of $C_1$, $I'$ and $I$ and observably equivalent to $\mathcal{I}$.

\david{Prove $C_2$.}
We now prove that $I'$ and $I$ are observably equivalent to $\mathcal{I}$ for inputs and outputs of $C_2$.
By construction of $I$, the inputs of $I$ are observably equivalent to $\mathcal{I}$.
Since $C_2$ is either functional or monotonic, and $C_2$ is independent of $C_1$, given the same input facts of $C_2$ in $I'$ and $I$, in any relation $r$ referenced in $C_2$, $f'$ exists in $r$ of $I'$ if and only if $f$ exists in $I$, where $f'$ equals $f$ except on time.
Consider output fact $f'$ derived at time $t'$ in $I'$ and $t$ in $I$.
In $I'$, $f'$ is buffered in \ded{outP} until the latest input fact of $C_2$ is acknowledged by $C_1$.
Since we assumed that $C_2$ is a state machine, all outputs of $C_2$ must have existence dependencies on its inputs, and the latest input fact $f_i'$ of $C_2$ in $I'$ must have time $t'$.
There must be a fact $f_c$ in \ded{canSeal} of $I'$ with $\pi_{T2}(f_c) = t'$.
By construction of $I$, $t = \pi_T(f_c)$.
The acknowledgement for $f_i$ must be sent at $t$ in $I'$ and arrive at some time $t_o > t$ at $C_2$, when $f'$ can be sent.
Therefore, the range of possible times of $f$ is $(t, \infty)$ and $(t_o, \infty)$ for $f'$, where the range of $f'$ is a sub-range of $f$, and any immediate consequence of $f'$ in $I'$ must be possible in $I$.

\section{Partitioning}

\subsection{Partitioning by co-hashing}
\label{app:partitioning-by-co-hashing}

\subsubsection{Mechanism}
\label{app:partitioning-by-co-hashing-mechanism}

After finding a distribution policy $D$ that partitions consistently with co-hashing, the partitioning rewrite routes input facts to nodes \ded{addr\_i} by injecting the distribution policy $D$.
We model $D$ as a relation in Dedalus, such that if $D(f) =$ \ded{addr\_i}, then we add the tuple \ded{D(..., addr\_i)}, where \ded{...} represents the values of $f$, excluding time:
We then apply the following rewrite:

\textbf{Rewrite: Redirection With Partitioning.} Given a rule $\varphi$ in another component $C'$ whose head relation $r$ is referenced in $C$:
\begin{itemize}
    \item Add the body term $D(\ldots, l')$ to $\varphi$, where $\ldots$ is bound to the variables in the head of the rule.
    \item Replace the location variable $v$ of the head of the rule with $l'$. Replace $v$ with $l'$ in \ded{delay} similarly.
\end{itemize}
Note that this rewrite differs from the Redirection rewrite (\Cref{sec:mutually-independent-decoupling}) in that the entire fact is used to determine the new destination, while the Redirection rewrite only considers the original destination.

\subsubsection{Proof}
\label{app:partitioning-by-co-hashing-proof}

Formally, we will prove that for each instance $I'$ over the component $C$ partitioned by co-hashing with distribution policy $D$, there exists $I$ over the original component $C$ such that (1) for output relations $r$, $I$ contains fact $f$ in $r$ if and only if $I'$ contains $f$, and (2) for input relations $r$ or relations referenced in $C$, $I$ contains $f$ in $r$ if and only if $I'$ contains $f'$, where $f$ and $f'$ share the same values except $\pi_L(f') = D(f)$.

First consider $\varphi'$ where the head relation $r$ of $\varphi'$ is an input of $C$, and $\varphi'$ is a rule of some other component $C'$.
Let $\varphi$ be the original rule.
Facts in the body of $\varphi$ and $\varphi'$ are the same in both $I'$ and $I$, and the partitioning rewrite only changes the location of $f'$ in the head of $\varphi'$ such that $\pi_L(f) = D(f)$.
Thus our claim holds for input relations of $C$.

Now consider $\varphi$ in $C$.
All relations in the body of $\varphi$ are referenced by $C$, by definition.

Consider synchronous or inductive $\varphi$, such that the head relation $r$ of $\varphi$ is referenced in $C$ as well; we need to prove that fact $f$ of $r$ in $T_\varphi(I)$ exists if and only if $f'$ exists in $T_{\varphi'}(I')$, where $\pi_L(f') = D(f')$.
Assume by induction that this holds for $I$ and $I'$.
$\varphi$ equals $\varphi'$ since the rewrite does not alter rules in $C$.
For relation $r_h = head(\varphi)$ and any $r_b \in body(\varphi)$, there must be some attributes $A_h,A_b$ of $r_h,r_b$ that share keys; otherwise $D$ cannot exist and we do not partition.
Let $f_b$ be the fact of $r_b$ in $T_\varphi(I)$, and $f_b'$ be the corresponding fact in $T_\varphi(I')$.
$D$ must partition consistently with co-hashing on $A_h,A_b$, so $D(f_h) = D(f_b)$ and $D(f_h') = D(f_b')$.
$\varphi$ is synchronous or inductive, so we must have $\pi_L(f_h') = \pi_L(f_b')$.
The inductive hypothesis states that $D(f_b') = \pi_L(f_b')$.
Since $D(f_h') = D(f_b')$ and $\pi_L(f_h') = \pi_L(f_b')$, we must have $D(f_h') = \pi_L(f_h')$.
Therefore the induction hypothesis holds for all instances $I$ and $I'$.

If $\varphi$ is asynchronous, relations $r \in body(\varphi)$ must be inputs or referenced relations.
By the proofs above, facts $f_1,f_2$ of any pair of relations $r_1,r_2 \in body(\varphi)$ are partitioned consistently with co-hashing on the attributes that share keys.
Therefore $D(f_1) = D(f_2)$, implying $\pi_L(f_1') = \pi_L(f_2')$ for the corresponding $f_1',f_2'$ which are otherwise identical to $f_1,f_2$.
Therefore any immediate consequence over these facts $f_1,f_2$ of $T_\varphi(I)$ can always correspond to some immediate consequence $T_\varphi(I')$ over $f_1',f_2'$, and vice versa.
The location of the head (the output relation) is unmodified, completing our proof.

Note that we do not separately prove correctness for aggregation and negation because they are covered by our proofs according to the definition of ``sharing keys'' in \Cref{sec:co-hashing-partitioning}.

\subsection{Partitioning with dependencies}
\label{app:partitioning-with-dependencies}

\subsubsection{Checks for dependencies}
\label{app:partitioning-with-dependencies-checks}
FDs in Dedalus can be created in three ways:
\begin{itemize}
    \item \textbf{EDB annotation.} For example, \ded{hash(M, H)} is the EDB relation that simulates the hash function $hash(m) = h$, so there is an FD $g: M \to H$ where $g(m) = hash(m)$.
    \item \textbf{Variable sharing.} For a relation $r$, if in all rules $\varphi$ with $r$ as the head, attributes $A,B$ of $r$ always share keys, then there is an FD from $A$ to $B$ and from $B$ to $A$, where $g(x) = x$.
    \item \textbf{Inheritance.} The heads of rules $\varphi$ can inherit functional dependencies from a combination of joined relations in the body of $\varphi$.
\end{itemize}
In the last case, to determine which dependencies are inherited, we perform the following analysis for each relation $r$ in the head of rule $\varphi$:
\begin{itemize}
    \item \textbf{Attribute-variable substitution.} Take the set of all FDs of all relations in the body of $\varphi$ and replace each domain/co-domain on attributes with their bound variables in $\varphi$.
    \item \textbf{Constant substitution.} If an attribute is bound to a constant instead of another variable, plug the constant into the FDs of that attribute. Now all FDs in $\varphi$ should be functions on variables.
    \item \textbf{Transitive closure.} Construct the transitive closure of all such FDs.
    \item \textbf{Variable-attribute substitution.} Replace each FD on variables with their bound attributes from $r$, if possible. The FDs that only contain attributes in $r$ are the possible FDs of $r$.
\end{itemize}

\david{Intersection between rules for FDs.}
Having described the process of extracting FDs for each relation $r$ at the head of each rule $\varphi$, we must determine which FDs hold \emph{across} rules.
Since the identification of FDs for any relation assumes that all dependent relations have already been analyzed, and Dedalus allows dependency cycles, FD analysis must be recursive.
We divide the process of identifying FDs for $r$ into two steps: union and intersection.
The union step recursively takes the union of generated dependencies in any rule $\varphi$ where $r$ is the head; the intersection step recursively removes FDs that are not generated in some rule $\varphi$ where $r$ is the head.

\david{Finding CDs.}
CDs can be similarly extracted using dependency analysis.
In the variable-attribute substitution step, instead of only retaining FDs where variables in the domain and co-domain can \emph{all} be replaced with attributes in $r$, retain FDs where \emph{any} variable can be replaced with attributes in $r$.
These are the CDs of $r$ and the relations $r$ joins with in $\varphi$; they describe how attributes of $r$ joins with other relations in $\varphi$.
The CDs that hold across rules can be identified with the intersection step for FDs.

\subsubsection{Mechanism}
\label{app:partitioning-with-dependencies-mechanism}
The rewrite mechanics are identical to those in \Cref{app:partitioning-by-co-hashing-mechanism}.

\subsubsection{Proof}
\label{app:partitioning-with-dependencies-proof}

The proofs are similar to those in \Cref{app:partitioning-by-co-hashing-proof}, assuming a CD $g$ exists over attributes $A,B$ of relations $r_1,r_2$ only if for any $f_1,f_2$ of $r_1,r_2$ in the same proof tree, we must have $\pi_A(f_1) = g(\pi_B(f_2))$.

\subsection{Partial partitioning}
\label{app:partial-partitioning}

\subsubsection{Mechanism}
\label{app:partial-partitioning-mechanism}

In order to replicate relations $r$ referenced in $C_1$, we create a new ``coordinator'' proxy node at \ded{addr'}, introduce a relation \ded{proxy}$(l, l')$, and populate \ded{proxy} with a tuple $(addr, addr')$.

\textbf{Rewrite: Replication.}
Given a rule $\varphi$ in another component $C'$ whose head relation $r$ is referenced in $C$:
\begin{itemize}
    \item If $r$ is referenced in $C_1$, add the body term \ded{proxy}$(l, l')$ to $\varphi$ and bind the location attribute of the head to $l$.
    \item Otherwise, apply the partitioning rewrite.
\end{itemize}

We describe the functionality of the proxy node but omit its implementation.
The proxy node acts as the coordinator in 2PC, where the partitioned nodes are the participants.
It receives input facts on behalf of $C_1$ as described above, assigns each fact a unique, incrementing order, and broadcasts them to each node through \ded{rVoteReq}.
The nodes freeze and reply through \ded{rVote}.
The proxy waits to hear from all the nodes, then broadcasts the message through \ded{rCommit}, which unfreezes the nodes.
We describe how the nodes are modified to freeze, vote, and unfreeze (only after receiving all previously voted-for values) below:

Add the following rules to $C$:
\begin{lstlisting}[language=Dedalus, float=false]
processedI(i,l,t') :- processedI(i,l,t), t'=t+1
maxProcessedI(max<i>,l,t) :- processedI(i,l,t)
maxReceivedI(max<i>,l,t) :- receivedI(i,l,t)
unfreeze(l,t) :- maxReceivedI(i,l,t), maxProcessedI(i,l,t), !outstandingVote(l,t)
\end{lstlisting}
As we show below, \ded{unfreeze(l,t)} will be appended to rules so they can only be executed when all previous replicated inputs have been processed.

For each relation $r$ referenced in $C_1$, replace $r$ with \ded{rSealed} and add the following rules:
\begin{lstlisting}[language=Dedalus, float=false]
# Send replicated messages to the proxy for ordering.
rVoteReq(...,l,t') :- rVoteReq(...,l,t), t'=t+1
rVote(l,...,l',t') :- rVoteReq(...,l,t), proxy(l,l'), delay((...,l,t,l'),t')
# The proxy sends rCommit when all partitions have sent rVote.
rCommit(i,...,l,t') :- rCommit(i,...,l,t), t'=t+1
receivedI(i,l,t) :- rCommit(i,...,l,t)
# Messages in rCommit are processed in the proxy-assigned order.
rSealed(next,...,l,t) :- maxProcessed(i,l,t), next=i+1, rCommit(next,...,l,t)
rSealed(0,...,l,t) :- !maxProcessed(i,l,t), rCommit(0,...,l,t)
processedI(i,l,t') :- rSealed(i,...,l,t), t'=t+1
outstandingVote(l,t) :- rVoteReq(...,l,t), !rCommit(i,...,l,t)
\end{lstlisting}

For each remaining relation $r$ in $C$, replace $r$ with \ded{rSealed}, and add the following rules:
\begin{lstlisting}[language=Dedalus, float=false]
r(...,l,t') :- r(...,l,t), t'=t+1, !unfreeze(l,t)
rSealed(...,l,t) :- r(...,l,t), unfreeze(l,t)
\end{lstlisting}

\subsubsection{Proof}
\label{app:partial-partitioning-proof}

Before stating our proof goal, we present terms to describe partitioned state.
For simplicity, we denote $I_{\overline{r},p,t}$ as the set of facts $f$ in $I$ where $f$ is a fact of a relation in $\overline{r}$, $D(f) = p$, and $\pi_T(f) = t$.
A relation $r$ is \textit{empty} at time $t$ in instance $I$ and node $p$ if there is no fact $f$ of $r$ in $I$ where $\pi_T(f) = t$ and $D(f) = p$.
The \textit{corresponding facts} of a replicated fact $f$ are all facts $f'$ where $f$ equals $f'$ except $\pi_L(f) \ne \pi_L(f')$.

Since input relations $r$ to $C_1$ now arrive at the proxy first before being forwarded to the partitioned nodes, the arrival time of input facts in the transformed component $C$ no longer corresponds to the processing time.
This poses a problem for our proof; in the original component $C$, the arrival time of input facts \emph{is} the processing time.
We relax this requirement through the observation that because messages are sent over an asynchronous network, any entity that sends and receives messages from $C$ can only observe the ``send'' time $t_s$ of each input fact $f$, where $t_s < \pi_T(f)$.
Intuitively, an input fact that is sent at time $t_s$ and in-network for $t$ seconds is processed identically to a fact sent at time $t_s$, in-network for $t'$ seconds, and buffered for $b$ seconds, so long as $t' + b = t$.

\david{Proof.}
We will prove that for each instance $I'$ over the partially partitioned component $C$, there exists $I$ over the original component $C$ and the observable instance $\mathcal{I}$ such that (1) $I$ and $I'$ are both observably equivalent to $\mathcal{I}$, and (2) for the set of relations $\overline{r}$ referenced in $C$ (and $\overline{r'}$ corresponding to $\overline{r}$ with \ded{rSealed} replacing $r$), for any time $t'$ and node $p$ where at least one input relation \ded{rSealed} of $C_1$ is not empty in $I'$ and contains fact $f'$, let $f$ be the fact in $I$ corresponding to $f'$ with the smallest time $t = \pi_T(f)$, and let $p = \pi_L(f')$.
We must have $I'_{\overline{r'},p,t'} = I_{\overline{r},p,t}$.
In other words, although each replicated input is delivered at different times at different nodes, the states of each node at their differing times of delivery correspond to the original state at a single time of input delivery.
(3) Similarly, if at least one input relation \ded{rSealed} of $C_2$ is not empty, then the same condition holds with $t' = t$, $p = D(f')$.

\david{If the state at input times $t'$ in $I'$ match state in $I$ at $t$ (claims 2 and 3) + output facts are identical, then $I'$ and $I$ are both observably equivalent to $\mathcal{I}$ (claim 1).}
First, observe that claim 1 holds if (a) each output fact $f$ of $C$ exists in $I'$ if and only if $f$ is also in $I$, and (b) claims 2 and 3 hold.
Let $\mathcal{I}$ be observably equivalent to $I$; each input fact $f$ of $r$ in $I$ exists if and only if there exists $f_s$ in $\mathcal{I}$ where $f$ equals $f_s$ except $\pi_T(f) > \pi_T(f_s)$.
By claims 2 and 3, there must be $f'$ in $I'$ where $f'$ equals $f$ except $\pi_T(f') \ge \pi_T(f)$ (since $\pi_T(f)$ is based on the smallest $\pi_T(f')$ across nodes), which implies $\pi_T(f') > \pi_T(f_s)$.
If output facts in $I$ and $I'$ are identical, then $I'$ must also be observably equivalent to $\mathcal{I}$.

\david{Claims 2 and 3 imply that output facts can be identical.}
Claims 2 and 3 imply that output facts in $I'$ are also possible in $I$.
Output relations in $C$ are assumed to have existence dependencies on inputs, so given the inputs $f'$ and $f$ above, if the range of possible times for any output fact $f_o'$ in $I'$ is $(\pi_T(f'), \infty)$, then the range of possible times for the same output facts $f_o$ in $I$ must be $(\pi_T(f), \infty)$.
Since $\pi_T(f') \ge \pi_T(f)$, so the range of possible output times of $f_o'$ must be a sub-range of $f_o$, and any $f_o'$ in $I'$ is possible in $I$.

Therefore it suffices to prove claims 2 and 3.

\david{Selecting input arrival times for $I$.}
We select the time of input facts in $I$ such that its inputs are observably equivalent to $\mathcal{I}$.
For facts $f'$ in $I'$ over relations \ded{rSealed} in $C$:
(1) If \ded{rSealed} is referenced in $C_1$, let $f_a'$ be the corresponding fact of $f'$ in $I'$ with the smallest time $t' = \pi_T(f')$, and let $f$ equal $f_a'$ except $\pi_L(f) =$ \ded{addr}; $f$ of the corresponding $r$ exists in $I$ if and only if $f'$ exists in $I'$.
(2) If \ded{rSealed} is referenced in $C_2$, let $f$ equal $f'$ except $\pi_L(f') =$ \ded{addr}; $f$ of the corresponding $r$ exists in $I$ if and only if $f'$ exists in $I'$.
By the partial partitioning rewrite, each fact in \ded{rSealed} is derived from a series of facts in \ded{r}, \ded{rVote}, \ded{rVoteReq}, then \ded{rCommit}.
The fact $f_r$ in $r$ in its proof tree must have an earlier time, which must be later than the send time of $f_r$.
Therefore, since the times of input facts $f$ in $I$ is set to the times of facts in \ded{rSealed}, $\pi_T(f)$ must be later than the send time of $f$, and the inputs of $I$ are observably equivalent to $\mathcal{I}$.

We now prove claim 2 and 3 by induction on facts with time $t_i$ in $I'$.

\david{Prove claim 2.}
We show that claim 2 holds for facts with time $t' = t_i+1$, assuming claims 2 and 3 hold up to $t_i$.
We will prove by induction of the inputs in $I'$ and $I$ up to time $t'$.
Let $f'$ be a fact in input relation \ded{rSealed} of $C_1$ in $I'$ with $t' = \pi_T(f')$, $p = \pi_L(f')$, and $f_a'$ be the fact in $I'$ corresponding to $f'$ with the smallest time $t = \pi_T(f_a')$.
Let $f$ equal $f_a'$ except $\pi_L(f) =$ \ded{addr}.
By definition of $I$ above, $f$ is an input fact in $I$.
By the partial partitioning rewrite, there is no other co-occurring input fact $f''$ in $I'$ with $\pi_T(f'') = t'$, since \ded{processedI} will not contain the index of \ded{rCommit} until the next timestep.
Therefore, $I$ does not contain any other input fact at $t$.
Let $f_1'$ be the most recent fact in any input relation $r_1$ in $I'$ with $t_1' = \pi_T(f_1')$, $t' > t_1'$, and $p = \pi_L(f_1')$, such that there does not exist $f_2'$ input fact with $t_2' = \pi_T(f_2')$ in $I'$ where $t' > t_2' > t_1'$.
Let $f_{1a}'$ correspond to $f_1'$ with the smallest time $t_1 = \pi_T(f_{1a}')$.
Let $f_1$ equal $f_{1a}'$ except $\pi_L(f_1) =$ \ded{addr}.
By definition of $I$ above, $f_1$ is an input fact in $I$.
By the inductive hypothesis, we have $I'_{\overline{r'},p,t_1'} = I_{\overline{r},p,t_1}$.

Since no input facts exist in $I'$ for node $p$ between $t_1'$ and $t'$, if we can show that no input facts exist in $I$ for partition $p$ between $t_1$ and $t$, then we can reuse the proof from \Cref{app:state-machine-decoupling-proof} combined with the partitioning proof from \Cref{app:partitioning-by-co-hashing-proof} to show that $I'_{\overline{r'},p,t'} = I_{\overline{r},p,t}$.

\begin{lemma}[Order consistency]
\label{lemma:order-consistency}
Given facts $f_1,f_2$ of an input relation of $C_1$ in $I$, $\pi_T(f_1)$ is less than $\pi_T(f_2)$ if and only if for each node $p$, for each pair of corresponding facts $f_1',f_2'$ in $I'$ where $\pi_L(f_1') = \pi_L(f_2') = p$, we have $\pi_T(f_1') < \pi_T(f_2')$.
\end{lemma}
\begin{proof}
Assume by contradiction that for some node $p$, $\pi_T(f_1') \ge \pi_T(f_2')$.
By the partial partitioning rewrite, we know that $\pi_T(f_1') = \pi_T(f_2')$ is impossible for input relations of $C_1$.
Let the order $O(f')$ of a fact $f'$ in an input relation \ded{rSealed} of $C_1$ in $I'$ be the value of the index attribute of its parent fact in \ded{rCommit}.
Then $\pi_T(f_1') > \pi_T(f_2')$ if and only if $O(f_1') > O(f_2')$, which holds across all nodes $p$ by the partial partitioning rewrite.
Now let $f_{1a}'$ be the corresponding fact of $f_1$ with the smallest time in $I'$, and $f_{2b}'$ for $f_2$.
By construction of $I$, $\pi_T(f_{1a}') = \pi_T(f_1)$ and $\pi_T(f_{2b}') = \pi_T(f_2)$.
Let $f_{2a}'$ be the corresponding fact of $f_2$ where $\pi_L(f_{2a}') = \pi_L(f_{1a}')$; $f_{1a}'$ and $f_{2a}'$ are the corresponding inputs for $f_1$ and $f_2$ on one specific node.
Define $f_{1b}'$ for $f_{2b}'$ similarly.
Since ordering is consistent across nodes, $O(f_1') > O(f_2')$ implies $O(f_{1a}') > O(f_{2a}')$, which implies $\pi_T(f_{1a}') > \pi_T(f_{2a}')$.
Since $\pi_T(f_1) < \pi_T(f_2)$, we know $\pi_T(f_{1a}') < \pi_T(f_{2b}')$, therefore $\pi_T(f_{2a}') < \pi_T(f_{2b}')$.
This contradicts the definition of $f_{2b}'$ as the fact in $I'$ corresponding to $f_2$ with the smallest time, since $f_{2a}'$ corresponds to $f_2$ with a smaller time.
Proof by contradiction.
\end{proof}

Assume for contradiction that there exists $f_2$ of $r_2$ in $I$, $t_2 = \pi_T(f_2)$, $t > t_2 > t_1$, such that the most recent fact of $I'$ does not correspond to the most recent fact of $I$.
Let $f_2'$ be the corresponding input fact in $I'$ with $t_2' = \pi_T(f_2')$, $p = \pi_L(f_2')$ such that $t' > t_2'$.

If $r_2$ and $r_1$ are both inputs of $C_1$, then by \Cref{lemma:order-consistency}, $t > t_2 > t_1$ implies $t' > t_2' > t_1'$, which contradicts the definition of $f_1'$ as the most recent input fact.

If $r_2$ is an input of $C_1$ and $r_1$ is an input of $C_2$, then $t > t_2$ implies $t' > t_2'$, and $t_1 = t_1'$.
Since $t_2 > t_1$ and $t_2' \ge t_2$ by construction of $I$, we have $t' > t_2' > t_1'$ which is again a contradiction.

Now consider $r_2$ as an input of $C_2$ such that $t_2' = t_2$.
If $r_1$ is an input of $C_1$, then let $f_{1a}'$ be the corresponding input fact in $I'$ with the smallest time $t_{1a}' = \pi_T(f_{1a}')$.
By definition, $t_1' \ge t_{1a}'$.
By construction of $I$, $t_{1a}' = t_1$, therefore $t > t_2 > t_1$ implies $t' > t_2' > t_{1a}'$.
By the partial partitioning rewrite, no input facts can arrive on partition $p$ between $t_{1a}'$ and $t_1'$, since the \ded{outstandingVote} relation will force input relations to buffer.
Therefore, given $t_2' > t_{1a}'$ and $t_1' \ge t_{1a}'$, we must have $t_2' > t_1' \ge t_{1a}'$.
Combined with $t' > t_2' > t_{1a}'$, we now have $t' > t_2' > t_1' \ge t_{1a}'$; $t_2' > t_1'$ contradicts the definition of $f_1'$ as the most recent input fact.

Otherwise, if $r_1$ is an input of $C_2$, then $t_1' = t_1$.
Then $t_2 > t_1$ implies $t_2' > t_1'$, which contradicts the definition of $f_1'$ as the most recent input fact.

We've proven that the most recent input fact $f_1'$ in $I'$ up to $t'$ corresponds to the most recent input fact $f_1$ in $I$ up to $t$, and we can reuse earlier proofs to prove claim 2.

\david{Prove claim 3.}
We show claim 3 holds, using the same variable definitions.
By the partial partitioning rewrite, we know that if an input relation \ded{rSealed} of $C_2$ is not empty, then input relations \ded{rSealed} of $C_1$ must be empty.
Let $f'$ be the input fact in $I'$ at time $t'$ with corresponding $f$ in $I$ at $t$.
Note that this proof can be generalized to a set of facts $\overline{f'}$ at time $t'$.
By construction of $I$, $t' = t$.
Again assuming the most recent input fact $f_1'$ of $r_1$ in $I'$ at time $t_1'$ corresponds to $f_1$ in $I$ at $t_1$, we show that there is no $f_2$ of $r_2$ in $I$ at $t_2$ where $t > t_2 > t_1$.

If $r_2$ and $r_1$ are both inputs of $C_1$, then $t_2 > t_1$ implies $t_2' > t_1'$ by \Cref{lemma:order-consistency}.
Since $t = t'$, so $t > t_2 > t_1$ implies $t' > t_2' > t_1'$, which contradicts the definition of $f_1'$ as the most recent input fact.

If $r_2$ is an input of $C_1$ and $r_1$ is an input of $C_2$, then $t_1 = t_1'$.
Since $t_2' \ge t_2$ by construction of $I$, $t_2 > t_1$ implies $t_2' > t_1'$, and since $t = t'$, $t > t_2 > t_1$ implies $t' > t_2' > t_1'$, a contradiction.

If $r_2$ is an input of $C_2$ and $r_1$ is an input of $C_1$, then $t_2 = t_2'$.
Since $t = t'$ and $t > t_2$, $t' > t_2'$.
$t' > t_2' > t_1'$ is a contradiction, and $t_2' \ne t_1'$ by the partial partitioning rewrite, so we must have $t' > t_1' > t_2'$.
Let $f_{1a}'$ correspond to $f_1$ in $I'$ with the smallest time $t_{1a}' = \pi_T(f_{1a}')$.
By construction of $I$, $t_{1a}' = t_1$.
By definition, $t_1' \ge t_{1a}'$.
By the partial partitioning rewrite, no input facts can arrive on node $p$ between $t_{1a}'$ and $t_1'$.
Therefore, given $t_1' \ge t_{1a}'$ and $t_1' > t_2'$, we must have $t_1' \ge t_{1a}' > t_2'$.
$t_{1a}' > t_2'$ implies $t_1 > t_2$, which contradicts our assumption that $t > t_2 > t_1$.

If $r_2$ and $r_1$ are both inputs of $C_2$, then $t_2 = t_2'$ and $t_1 = t_1'$.
Combined with $t = t'$, $t > t_2 > t_1$ implies $t' > t_2' > t_1'$ which is a contradiction.

By induction on previous facts of input relations, we also have $I'_{\overline{r'},p,t_1'} = I_{\overline{r},p,t_1}$.
Since inputs are the same between $I'$ and $I$ at time $t'$ and $t$, we can similarly use the proofs from \Cref{app:partitioning-by-co-hashing-proof,app:state-machine-decoupling-proof} to prove claim 3, completing the proof of correctness.

\subsection{Partitioning sealing}
\label{app:partitioning-sealing}

\subsubsection{Sealing}
\label{sec:sealing}

Sealing~\cite{blazes} is a syntactic sugar we introduce to simulate sending multiple output facts in a single asynchronous message.
Sealed relations can be partially partitioned so each partition can compute and send its own fraction of the sealed outputs.

\david{Syntax.}
Syntactically, \ded{seal} is added to the head of a rule $\varphi$ using aggregation syntax.
We demonstrate how to seal the relation $r$ in the output relation \ded{out} of $C$ below, where $s$ and $u$ are additional illustrative relations:
\begin{lstlisting}[language=Dedalus, float=false, mathescape=true]
# Send rule on component $C$.
out(seal<r>,a,l',t') :- r(...,l,t), s(a,l,t), dest(l',l,t), delay((...,a,l,t,l'),t')
# Receive rule on component $C'$.
u(...,a,l,t) :- out(...,a,l,t)
\end{lstlisting}

This desugars into the following, where \ded{out} is replaced with \ded{outSealed} in $C'$:
\begin{lstlisting}[language=Dedalus, float=false, mathescape=true]
# Component $C$.
rCount(count<...>,a,l,t) :- r(...,l,t), s(a,l,t)
outCount(c,a,l',t') :- rCount(c,a,l,t), s(a,l,t), dest(l',l,t), delay((c,a,l,t,l'),t') |\label{line:sealing-out-count}|
out(...,a,l',t') :- r(...,l,t), s(a,l,t), dest(l',l,t), delay((...,a,l,t,l'),t') |\label{line:sealing-out}|
# Component $C'$.
outReceived(count<...>,a,l,t) :- out(...,a,l,t) |\label{line:sealing-received}|
sealed(a,l,t) :- outReceived(c,a,l,t), outCount(c,a,l,t) |\label{line:sealing-sealed}|
u(...,a,l,t) :- out(...,a,l,t), sealed(a,l,t)
# Only persist until sealed.
out(...,a,l,t') :- out(...,a,l,t), !sealed(a,l,t), t'=t+1
outCount(c,a,l,t') :- outCount(c,a,l,t), !sealed(a,l,t), t'=t+1
\end{lstlisting}

\david{Why sugaring helps us optimize.}
The sugared syntax for sealing guarantees that the relations \ded{rCount}, \ded{outCount}, \ded{outReceived}, etc are only used as described above, and correctness is preserved as long as after partial partitioning, the facts in \ded{outSealed} match the facts originally in \ded{out}.

\subsubsection{Mechanism}
\label{app:partitioning-sealing-mechanism}

Sealing can be partitioned through dependency analysis on the sugared syntax; among the rules $\varphi$ introduced in $C$, all the joins in the body of $\varphi$ were already present in the sugared syntax.

If \ded{out} cannot be partitioned with dependencies, and it has an existence dependency on some input relation \ded{in} of $C_1$, then the rewrite is as follows:

\textbf{Rewrite: Partitioning Sealing.}
After completing the partial partitioning rewrite over the sugared syntax, perform the following rewrites.
On $C$, replace \Cref{line:sealing-out-count,line:sealing-out} with the following, assuming \ded{inCommit} is defined for \ded{in} as specified in \Cref{app:partial-partitioning-mechanism}:
\begin{lstlisting}[language=Dedalus, float=false]
outCount(l,i,c,a,l,t) :- rCount(c,a,l,t), inCommit(i,...,l,t), dest(l',l,t), delay((l,i,c,a,l,t,l'),t') |\label{line:sealing-out-count}|
out(i,...,a,l',t') :- r(...,a,l,t), s(a,l,t), inCommit(i,...,l,t), dest(l',l,t), delay((i,...,a,l,t,l'),t')
\end{lstlisting}

On $C'$, introduce the relation \ded{numPartitions}$(n)$ and populate with it the number of nodes $addr_i$.
Replace \Cref{line:sealing-received,line:sealing-sealed} with the following code.
\begin{lstlisting}[language=Dedalus, float=false]
outReceived(i,count<...>,a,l,t) :- out(i,...,a,l,t)
# Sum the expected messages from all partitions.
outCountSum(i,sum<c>,a,l,t) :- outCount(p,i,c,a,l,t)
# Check if all partitions have sent their counts.
outCountPartitions(count<p>,i,a,l,t) :- outCount(p,i,c,a,l,t)
sealed(a,l,t) :- outReceived(i,c,a,l,t), outCountSum(i,c,a,l,t), outCountPartitions(n,i,a,l,t), numPartitions(n)
\end{lstlisting}

\subsubsection{Proof}
\label{app:partitioning-sealing-proof}

Unlike the proof in \Cref{app:partial-partitioning-proof}, we cannot reuse the proof in \Cref{app:partitioning-by-co-hashing-proof} because \ded{out} is technically neither fully or partially partitioned; instead \ded{outCount} is partially partitioned and the logic of $C'$ is modified.
Instead, we show that for instance $I'$ after transformation, $r =$ \ded{sealedOut}, for all time $t$, there exists $I$ such that $I'_{r,t} = I_{r,t}$.
The proof of correctness for all other relations in $C$ is covered by \Cref{app:partial-partitioning-proof}, as $r$ is an output relation so no other relations in $C$ are dependent on $r$.

\david{The state across replicas when the replicated input arrived on $I'$ corresponds to some state in $I$ when the input arrived.}
Since $r$ has an existence dependency on some input \ded{rSealed} of $C_1$, we only need to consider the times when \ded{rSealed} is not empty.
Consider time $t'$ in $I'$ where \ded{rSealed} contains some fact $f'$.
By partial partitioning, there is no other fact in \ded{rSealed} with the same time as $f'$.
Consider all corresponding facts of $f'$ for each node $p$, and their corresponding times $t_p'$.
By the proof in \Cref{app:partial-partitioning-proof}, there exists a time $t$ in $I$ such that for each $t_p'$, $I'_{\overline{r'},p,t_p'} = I_{\overline{r},p,t}$.

\david{Same potential output times.}
Let $\varphi$ be any rule in $C$ with $r$ in its head.
The body of $\varphi$ are referenced in $C$ and included in $\overline{r'}$ in $I'$ and $\overline{r}$ in $I$, respectively, and since those relations contain the same facts at times $t_p'$ and $t$, the will evaluate to the same output facts $f'$ and $f$, except the range of possible times for $f'$ is $(t_p', \infty)$ while the range is $(t, \infty)$ for $f$.
By construction of $I$, $t_p' \ge t$, so the range of possible times for $f'$ is a subset of the range for $f$, and any immediate consequence producing $f'$ in $I'$ is possible in $I$.
Note that the addition of the body term \ded{inCommit} does not affect the immediate consequence, since we assumed that $r$ already has an existence dependency on \ded{in}.
Similar logic applies for the evaluation of \ded{outCount}, completing our proof.

\section{Non-linearizable execution over partitioned acceptors}
\label{app:non-linearizable-acceptors}

\CompPaxos{} partitions acceptors without introducing coordination, allowing each node to hold an independent ballot.
In contrast, \AutoPaxos{} can only partially partition acceptors and must introduce coordinators to synchronize ballots between nodes, because our formalism states that the nodes' ballots together must correspond to the original acceptor's ballot.
Proposers in \CompPaxos{} can become the leader after receiving replies from a quorum of any $f+1$ acceptors for each set of $n$ nodes; the nodes across quorums do not need to correspond to the same acceptors.
In contrast, the $n$ nodes of each acceptor in \AutoPaxos{} represent one original acceptor, so proposers in \AutoPaxos{} become the leader after receiving replies from all $n$ nodes of a quorum of $f+1$ acceptors.
Crucially, by allowing the highest ballot held at each node to diverge, \CompPaxos{} can introduce non-linearizable executions that remain safe for Paxos, but are too specific to generalize.

We first define what it means for a Paxos implementation to be linearizable.
A \ded{p1a} and its corresponding \ded{p1b} correspond to a request and matching response in the history.
For an implementation of Paxos to be linearizable, the content of each \ded{p1b} must be consistent with its matching \ded{p1a} taking effect some time between the \ded{p1a} arrival time and the \ded{p1b} response time.
The same statements hold for \ded{p2a} and matching \ded{p2b} messages.
Since \ded{p1a} and \ded{p1b} messages are now sent to each node in \CompPaxos{}, we must modify the definition of linearizability for \CompPaxos{} accordingly.
Assume that a \ded{p1a} arriving at acceptor $a$ in Paxos corresponds to the arrival of the same \ded{p1a} messages at all nodes of $a$ in \CompPaxos{}, and a matching \ded{p1b} arriving at proposer $p$ in Paxos corresponds to the arrival of all matching \ded{p1b} messages at $p$ in \CompPaxos{}.

\begin{figure}[t]
    \centering
    \includegraphics[width=0.5\linewidth]{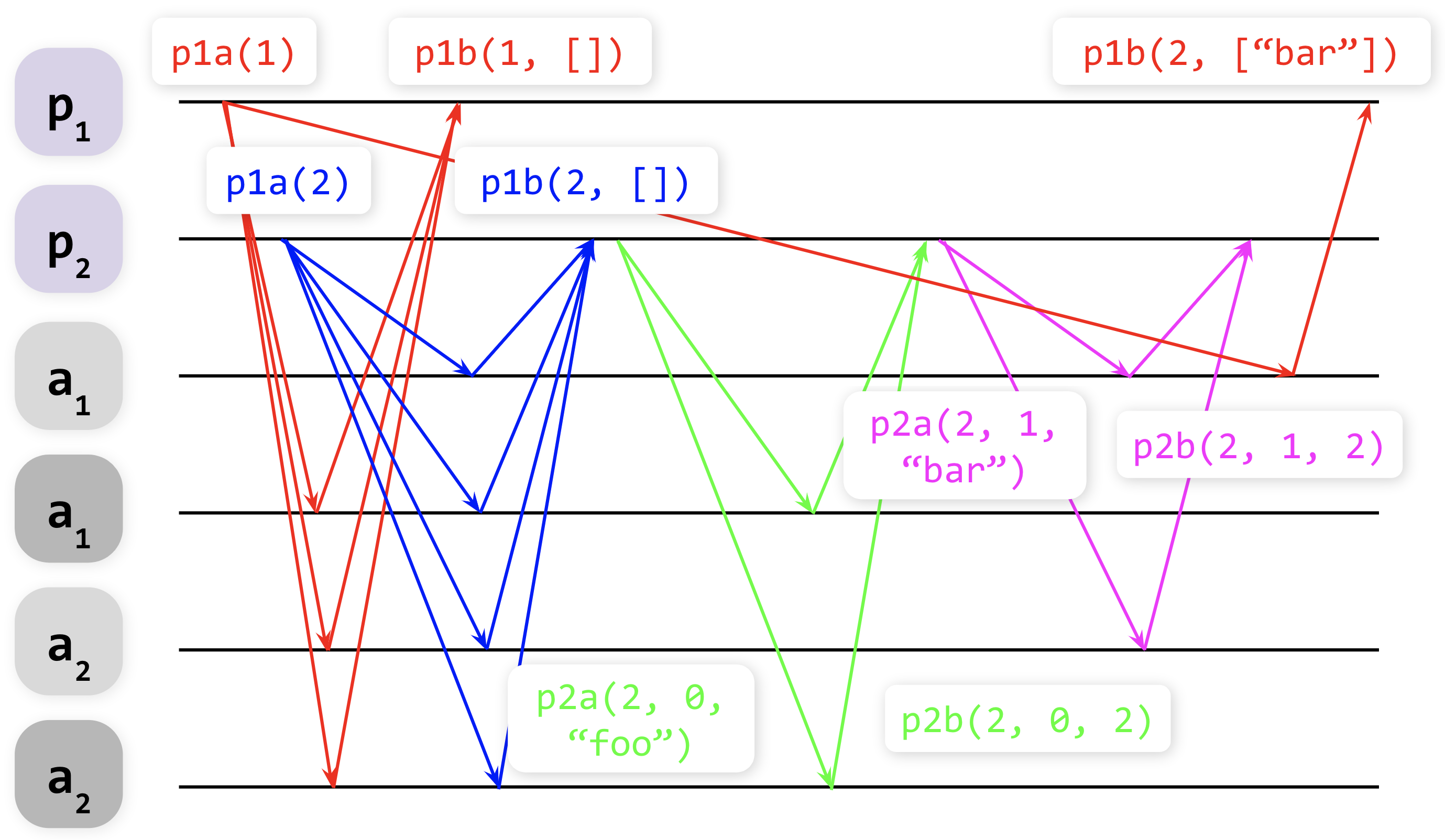}
    \caption{Non-linearizable execution of \CompPaxos{}. Acceptor $a_3$ is excluded for simplicity. Lighter-gray acceptors belong to partition 1, and darker-gray acceptors to partition 2. Each set of requests and matching responses are a different color. \kaushik{why not use the same gold colors used earlier for partitions?}}
    \label{fig:non-linearizability}
\end{figure}

Now consider the execution shown in \Cref{fig:non-linearizability}:
\begin{enumerate}
    \item Proposer $p_1$ broadcasts \ded{p1a} with ballot $1$. It arrives all acceptors except partition $1$ of acceptor $a_1$. The other acceptors return \ded{p1b} with ballot $1$.
    \item Proposer $p_2$ broadcasts \ded{p1a} with ballot $2$. It arrives at all partitions of every acceptor, which return \ded{p1b} with ballot $2$. Proposer $p_2$ is elected leader.
    \item Proposer $p_2$ sends \ded{p2a} with ballot $2$, message ``foo'', and slot $0$ to partition $2$ of every acceptor, which return \ded{p2b} with ballot $2$. ``foo'' is committed.
    \item Proposer $p_2$ then sends \ded{p2a} with ballot $2$, message ``bar'', and slot $1$ to partition $1$ of every acceptor, which return \ded{p2b} with ballot $2$. ``bar'' is committed.
    \item Proposer $p_1$'s \ded{p1a} finally arrives at partition $1$ of $a_1$, which returns \ded{p1b} with ballot $2$, containing ``bar'' in its log. Proposer $p_1$ merges the \ded{p1b}s it has received and concludes that the log contains only ``bar''.
\end{enumerate}

The execution is non-linearizable: by reading ``bar'', $p_1$'s \ded{p1b} must happen-after the write of ``bar'', which happens-after the write of ``foo'', but $p_1$ does not read ``foo'', so it must happen-before the write of ``foo''.
Why is \CompPaxos{} correct despite allowing such non-linearizable executions?

Non-linearizable reads of the log are only possible in Paxos when a proposer fails leader election, in which case, \textit{the log is discarded} and the proposer tries again.
The non-linearizable log is never used.
Intuitively, because phase 1 (leader election) quorums must intersect, in order for a proposer $p_1$ to read a write from $p_2$ that occurred while $p_1$ was attempting leader election, $p_2$ must have completed leader election, overlapping in at least 1 acceptor with $p_1$ and preempting $p_1$.
Thus $p_1$ will fail to become the leader and the log it receives in \ded{p1b} does not matter.

These differences stem from rewrites that are specific to Paxos and require an in-depth, global understanding of the protocol.
By design, our \ourApproach{} framework is protocol-agnostic and considers only local rewrites.
In contrast, \CompPaxos{} is able to admit rewrites that introduce non-linearizable executions as it can prove that the results of these executions are never used.

\jmh{Say all this up front. ``\CompPaxos{} was described as applying simple ideas for scaling Paxos, but in at least one regard is undertook a subtle global optimization that is not possible when adhering to rewriting as we did. Here we explain this optimization of \CompPaxos in detail.''}

\jmh{What am I supposed to take away from all these details? This feels unhelpful to my understanding of this paper. But the whole situation leaves me wondering about Figure 10. I don't know which of Michael's optimizations are turned on in your experiments -- you can easily help me with that in the text. I also don't know what to attribute to Scala vs Dedalus and what to attribute to Manual vs Automatic optimizations. That could be solved by hand-writing all Michael's optimizations in Dedalus, and redoing the ablation, which would be quite helpful. It would be nifty to show that decoupling and partitioning were all Michael really needed, but it would be equally interesting to see Michael's other optimizations adding a non-trivial boost -- challenge for future papers.}}{}

\tr{}{
\received{July 2023}
\received[accepted]{November 2023}}

\end{document}